\documentclass[preprint,nopreprintline,12pt]{elsarticle}
\usepackage[T1]{fontenc}
\usepackage{microtype}	                    
\usepackage{babel}
\usepackage[
  colorlinks=true,
  breaklinks=true
]{hyperref}                                 
\usepackage{amsmath}
\usepackage{amsthm}
\newtheorem{theorem}{Theorem}
\newtheorem{lemma}{Lemma}
\newtheorem{corollary}{Corollary}
\theoremstyle{definition}
\newtheorem{definition}{Definition}
\usepackage{amssymb}
\usepackage{algorithm}                      
\usepackage[noEnd=true]{algpseudocodex}     
\algrenewcommand\algorithmicrequire{\textbf{Input: }}
\algrenewcommand\algorithmicensure{\textbf{Output: }}
\usepackage{booktabs,tabularx}
\DeclareMathOperator{\win}{win}
\DeclareMathOperator{\conv}{conv}
\DeclareMathOperator{\ext}{ext}
\DeclareMathOperator{\DAG}{DAG}
\DeclareMathOperator{\suc}{succ}
\newcommand{\graphicaltest}{\Xi}
\newcommand{\cc}[1]{\operatorname{\mathsf{#1}}}
\newcommand{\decisionproblem}[3]{%
\begin{center}
  \begin{tabularx}{\columnwidth}{@{}lX@{}}
    \toprule
    \multicolumn{2}{@{}l@{}}{\textsc{#1}}\tabularnewline
    \midrule
    \bfseries Input: & #2 \\
    \bfseries Question: & #3 \\
    \bottomrule
  \end{tabularx}
\end{center}}
\usepackage{bm}
\usepackage{tikz}
\usetikzlibrary{decorations.pathmorphing}
\newlength{\vertexradius}
\def\radiusmain{1}
\tikzset{ arc/.style={thick,-stealth, shorten <=\vertexradius, shorten >=\vertexradius} }
\tikzset{ arcrev/.style={thick,stealth-, shorten <=\vertexradius, shorten >=\vertexradius} }
\newcommand{\arc}[0]{\!\raisebox{0.6ex}{\tikz[baseline]{\draw[arc] (0,0) -- ++(13pt,0);}}\!}
\usepackage[caption=false]{subfig}
\journal{Special Issue on Quantum Physics and Logic}
\begin{document}
\begin{frontmatter}
\title{Graphical Tests of Causality}
\author[usi,figa]{Ämin Baumeler} 
\author[ulb]{Eleftherios-Ermis Tselentis} 
\author[usi,figa]{Stefan Wolf} 
\affiliation[usi]{
  organization={Faculty of Informatics, Università della Svizzera italiana},
  addressline={\\Via la santa 1},
  city={Lugano-Viganello},
  postcode={6962},
  country={Switzerland}
}
\affiliation[figa]{
  organization={Facoltà indipendente di Gandria},
  addressline={Lunga scala 1},
  city={Gandria},
  postcode={6978},
  country={Switzerland}
}
\affiliation[ulb]{
  organization={Centre for Quantum Information and Communication, École polytechnique de Bruxelles, CP 165, Université libre de Bruxelles},
  addressline={50 av.\ F.D.\ Roosevelt},
  city={Brussels},
  postcode={1050},
  country={Belgium}
}

\begin{abstract}
  Bell inequalities limit the possible observations of non-communicating parties.
  Here, we present analogous inequalities for \emph{any number of communicating parties} under the causal constraints of
  \emph{static causal order,}
  \emph{definite causal order,}
  and \emph{bi-causal order.}
  All derived inequalities are remarkably simple.
  They correspond to upper bounds on the winning chance in \emph{graphical games:}
  Given a specific directed graph over the parties, the parties are challenged to communicate along a randomly chosen arc.
  In the case of \emph{definite causal order,} every game that we find is specified by a \emph{kefalopoda digraph.}
  Based on this we define \emph{weakly causal correlations} as those that satisfy all kefalopoda inequalities.
  We show that the problem of deciding whether some correlations are weakly causal is solvable in \emph{polynomial time in the number of parties.}
\end{abstract}

\end{frontmatter}

\section{Introduction, and Results}
Bell tests~\cite{bell1964,brunner2014} are central for separating quantum capabilities from classical ones.
If some observations at spacelike-separated regions violate a~Bell inequality, then a classical description of these observations is untenable:
They cannot be reproduced using a predefined ``lookup table'' shared among those regions.
A prime example~\cite{bell1964,clauser1969} of a Bell inequality,
which we explain further in Fig.~\ref{fig:bellpr}, is\footnote{%
  We use uppercase Latin letters for random variables and their lowercase versions for their value.
}%
\begin{align}
  \Pr[A\oplus B = XY] \leq 3/4
  \,.
  \label{eq:bell}
\end{align}
A paramount strength of this approach is that observed data only is used.
The specifics of the underlying theory are irrelevant.
Beyond their fundamental significance, Bell inequalities find applications in, e.g.,
self-testing of quantum devices~\cite{mayers1998,supic2020},
device-independent cryptography~\cite{bhk05,colbeckphd2009,scarani12,vaziranividick14,reviewdi23},
certification of randomness~\cite{colbeckphd2009,pironio2010,aaronson2023},
and classical proofs of quantum computation~\cite{kalai2023,broadbent2025}.
\begin{figure}
  \centering
  \subfloat[\label{subfig:bell}]{%
    \definecolor{sepia}{RGB}{224, 201, 166}
    \begin{tikzpicture}
      \foreach \party/\x/\O/\I in {Alice/-2/$a$/$x$, Bob/2/$b$/$y$} {
        \pgfmathsetmacro{\lx}{\x-0.5}
        \pgfmathsetmacro{\rx}{\x+0.5}
        \pgfmathsetmacro{\top}{-0.5}
        \pgfmathsetmacro{\bot}{-1.7}
        \pgfmathsetmacro{\nr}{7}
        \pgfmathsetmacro{\step}{(\top-\bot)/\nr}
        \pgfmathsetmacro{\tstep}{2*\step/3}
        \node at (\x,0) {\textbf{\party}};
        \fill[sepia!80,rounded corners=1pt] (\lx, \bot) rectangle (\rx, \top);
        \foreach \y in {1, ..., \nr} {
          \pgfmathsetmacro{\yy}{\y * \step}
          \draw[brown,decorate,decoration={snake,amplitude=0.7,segment length=4}] (\lx+0.05, \bot+\yy-\tstep) -- (\rx-0.05, \bot+\yy-\tstep);
        }
        \pgfmathsetmacro{\yc}{(\top+\bot)/2}
        \pgfmathsetmacro{\alen}{0.2}
        \draw[thick,->] (\rx,\yc) -- ++(\alen,0) node[right] {\O};
        \draw[thick,<-] (\lx,\yc) -- ++(-\alen,0) node[left] {\I};
      }
      \draw[dashed,gray] (0,0.3) -- (0,-2.25);
    \end{tikzpicture}
  }
  \hfill
  \subfloat[\label{subfig:pr}]{%
    \begin{tikzpicture}
      \node at (-2,0) {\textbf{Alice}};
      \node at ( 2,0) {\textbf{Bob}};
      \draw[fill=black,radius=0.1]
        (-2,-0.6) circle node[left] {$0$}
        -- node[midway,above] {$=$}
        ( 2,-0.6) circle node[right] {$0$}
        -- node[very near end,above] {$=$}
        (-2,-1.6) circle node[left] {$1$}
        -- node[midway,below] {$\neq$}
        ( 2,-1.6) circle node[right] {$1$}
        -- node[very near start,above] {$=$}
        (-2,-0.6);
    \end{tikzpicture}
  }
  \caption{%
    (a)~Alice and Bob are situated at distant locations without the possibility to communicate and share some resource.
    Alice computes the outcome~$a\in\{0,1\}$ from her measurement setting~$x\in\{0,1\}$.
    Similarly, Bob computes~$b\in\{0,1\}$ from~$y\in\{0,1\}$.
    If the resource is a ``lookup table,'' i.e., a classical object such as a book or a shared random variable, then their correlations cannot violate the Bell inequality~\eqref{eq:bell}.
    In contrast, if the resource is some quantum state, then they may violate that inequality.
    (b)~The predicate of the Bell inequality~\eqref{eq:bell} requires both outcomes to be different in the case~\mbox{$x=y=1$}, and equal otherwise.
    An immediate logical consequence of this is that no ``lookup table'' exist that satisfies the predicate,
    i.e., the constraint~\mbox{$a(0)=b(0)=a(1)\neq b(1)=a(0)$} is unsatisfiable~\cite{abramsky2012}, and at most three-out-of-four (in)equalities may be upheld simultaneously.
  }
  \label{fig:bellpr}
\end{figure}
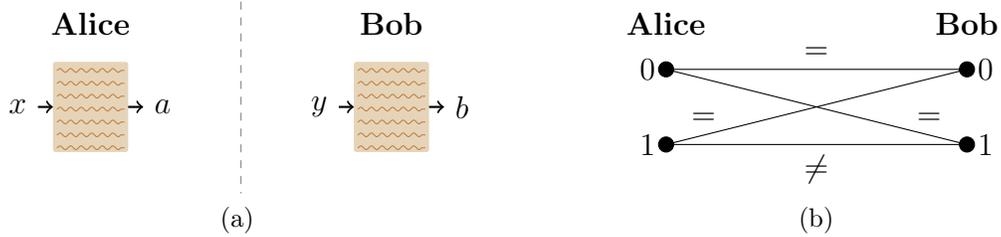

Recently, similar inequalities were derived for ``causality.''
Both, Bell inequalities and these inequalities for ``causality'' operate on the observed statistics only.
The difference lies in the causal assumptions.
In the Bell case, any communication is forbidden.
For~$n$ parties~$[n]:=\{0,1,\dots,n-1\}$ with the outcomes~$\underline A:=(A_i)_{i\in[n]}$ and the settings~$\underline X:=(X_i)_{i\in[n]}$, this means that the correlations decompose as
\begin{align}
  P_{\underline A|\underline X}
  =
  \sum_\lambda
  P_\Lambda(\lambda)
  \prod_{i\in[n]}
  P_{A_i | X_i, \Lambda=\lambda}
  \,;
  \label{eq:local}
\end{align}
there exists a probabilistic ``lookup table''~$\Lambda$, commonly known as ``hidden variable,''
through which each party~$i$ \emph{separately} derives the outcome~$a_i$ from the setting~$x_i$.

If, in contrast to the Bell case, we allow the parties to communicate, then the set of possible correlations enlarges.
With unrestricted communication, any distribution~$P_{\underline A|\underline X}$ is achievable.\footnote{%
  A trivial communication scheme to implement any distribution is guided by the decomposition~$P_{\underline A|\underline X}=P_{A_0|\underline X}P_{A_1|\underline X,A_0}P_{A_2|\underline X,A_0,A_1}\dots$:
  Every party~$i$ broadcasts its setting~$x_i$ to all parties, then party~$0$ computes~$a_0$ and broadcasts it, after which party~$1$ computes~$a_1$ and broadcasts it, and so forth.
}
Here, we focus on the following three causal constraints, which restrict communication in natural ways:
\begin{enumerate}
  \item \textbf{Static causal order~\cite{baumeler2014,oreshkov2016,tselentis2023}.}
    The parties are arranged along a~\emph{partial order,} i.e., all communication must respect that partial order.
    This is the most restrictive constraint we consider.
    Recently, this constraint was used to derive ``Bell tests'' of relativity under which deviations from special relativity can be detected~\cite{tselentis2023,tselentis2024arxiv}.
  \item \textbf{Definite causal order~\cite{ocb2012,baumeler2014,oreshkov2016,branciard2015}.}
    This constraint relaxes the previous one by allowing each party to arrange the causal relation among the parties within its future.
    This relaxation is motivated by general relativity, where \emph{anything within the future lightcone} --- hence also the causal order --- may be influenced.
  \item \textbf{Bi-causal order~\cite{abbott2017,baumeler2022}.}
    The only restriction is that one group of parties cannot communicate to any party outside of the group.
    A~violation of this constraint captures the \emph{genuinely multi-party character} of correlations incompatible with any causal order.
    This constraint is analogous to bi-locality~\cite{svetlichny1987} and other multi-party constraints~\cite{gallego2012,bancal2013} in the field of Bell nonlocality.
\end{enumerate}

\subsection{Results}
The Bell constraint (Eq.~\eqref{eq:local}), as well as the constraints of ``causality'' studied here, restrict the correlations~$P_{\underline A|\underline X}$ to belong to some convex polytope.
In a previous article~\cite{tselentis2023}, we focused on \emph{static causal order} in a simple scenario for any number of parties.
Here, we extend these studies to the constraints of \emph{definite causal order} as well as \emph{bi-causal order,}
and derive some of the facet-defining inequalities of the corresponding convex polytopes.
As it is the case for static causal order~\cite{tselentis2023}, we find that all derived tests are \emph{graphical:}
They admit a natural representation as directed graphs (digraphs).
In the following table, we summarize the correspondences between the causal constraints and digraphs:
\begin{table}[h!]
  \begin{center}
    \begin{tabular}[c]{|l|l|}
      \hline
      \textbf{Causal Constraint}
      &
      \textbf{Graphical Games}
      \\
      \hline
      Static causal order~~\hfill(Def.~\ref{def:static})
      &
      $k$-cycle, $k$-fence, $k$-Möbius~~\hfill(Cor.~\ref{cor:cyclefencemobius})
      \\
      Definite causal order~~\hfill(Def.~\ref{def:causal})
      &
      kefalopoda~~\hfill(Cor.~\ref{cor:kefalopoda})
      \\
      Bi-causal order~~\hfill(Def.~\ref{def:bicausal})
      &
      minimally strong~~\hfill(Cor.~\ref{cor:minimal})
      \\
      \hline
    \end{tabular}
  \end{center}
\end{table}

The \emph{graphical game} represented by some digraph~$D$ is straightforward:
Each vertex~$v\in\mathcal V(D)$ denotes a party, and given a random arc~$i\arc j\in\mathcal A(D)$, party~$i$ must communicate a bit to party~$j$.
In Fig.~\ref{fig:graphs}, we illustrate such games.
\begin{figure}
  \subfloat[\label{subfig:moebius}]{%
    \begin{tikzpicture}
      \def\tot{7}
      \pgfmathsetmacro{\angledelta}{360/\tot}
      \def\rl{1.5}
      \foreach \x in {1, 2, 3, ..., \tot} {
        \pgfmathsetmacro{\angle}{(\x-1) * \angledelta + \angledelta/2 + 90}
        \pgfmathsetmacro{\ra}{iseven(\x) ? \radiusmain : \rl}
        \pgfmathsetmacro{\rb}{isodd(\x) ? \radiusmain : \rl}
        \fill (\angle:\radiusmain) circle (\vertexradius);
        \fill (\angle:\rl) circle (\vertexradius);
        \draw[arc] (\angle:\ra) -- (\angle:\rb);
        \ifnum\x=1
          \draw[arc] (\angle:\rb) arc[start angle=\angle, delta angle=\angledelta,radius=\rb];
          \draw[arc] (\angle:\rb) to[out=\angle-90,in=180] (\angle-\angledelta:\ra);
        \else\ifnum\x=\tot
          \draw[arc] (\angle:\rb) to[out=\angle+90,in=0] (\angle+\angledelta:\ra);
          \draw[arc] (\angle:\rb) arc[start angle=\angle, delta angle=-\angledelta,radius=\rb];
        \else
          \draw[arc] (\angle:\rb) arc[start angle=\angle, delta angle=\angledelta,radius=\rb];
          \draw[arc] (\angle:\rb) arc[start angle=\angle, delta angle=-\angledelta,radius=\rb];
        \fi \fi
      }
    \end{tikzpicture}
  }
  \hfill
  \subfloat[\label{subfig:kefalopoda}]{%
    \begin{tikzpicture}
      \def\tot{5}
      \pgfmathsetmacro{\angledelta}{360/\tot}
      \foreach \x in {1, 2, 3, ..., \tot} {
        \pgfmathsetmacro{\angle}{(\x-1) * \angledelta + \angledelta/2 + 90}
        \fill (\angle:\radiusmain) circle (\vertexradius);
        \draw[arc] (\angle:\radiusmain) arc[start angle=\angle, delta angle=-\angledelta,radius=\radiusmain];
      }
      \pgfmathsetmacro{\angle}{(2-1) * \angledelta + \angledelta/2 + 90}
      \fill (\angle:\radiusmain) -- ++(220:1) coordinate (p) circle (\vertexradius);
      \draw[arc] (\angle:\radiusmain) -- (p);
      \fill (\angle:\radiusmain) -- ++(260:1) coordinate (p) circle (\vertexradius);
      \draw[arc] (\angle:\radiusmain) -- (p);
      \pgfmathsetmacro{\angle}{(4-1) * \angledelta + \angledelta/2 + 90}
      \fill (\angle:\radiusmain) -- ++(300:0.7) coordinate (p) circle (\vertexradius);
      \draw[arc] (\angle:\radiusmain) -- (p);
      \fill (p) -- ++(210:0.5) coordinate (q) circle (\vertexradius);
      \draw[arc] (p) -- (q);
      \fill (p) -- ++(300:0.5) coordinate (r) circle (\vertexradius);
      \draw[arc] (p) -- (r);
    \end{tikzpicture}
  }
  \hfill
  \subfloat[\label{subfig:minstrong}]{%
    \begin{tikzpicture}
      \def\tot{3}
      \pgfmathsetmacro{\angledelta}{360/\tot}
      \foreach \x in {1, 2, 3, ..., \tot} {
        \pgfmathsetmacro{\angle}{(\x-1) * \angledelta + \angledelta/2 - 90}
        \fill (\angle:\radiusmain) circle (\vertexradius);
        \draw[arc] (\angle:\radiusmain) arc[start angle=\angle, delta angle=-\angledelta,radius=\radiusmain];
        \fill (\angle:\radiusmain) -- ++(\angle:\radiusmain) coordinate (a) circle (\vertexradius);
        \draw[arc] (\angle:\radiusmain) to[bend left] (a);
        \draw[arc] (a) to[bend left] (\angle:\radiusmain);
      }
    \end{tikzpicture}
  }
  \caption{%
    Examples of three graphical games of causality:
    (a)~a Möbius digraph to test \emph{static causal order,}
    (b)~a kefalopoda digraph to test \emph{definite causal order,}
    and (c)~a minimally strong digraph to test \emph{bi-causal order.}
  }
  \label{fig:graphs}
\end{figure}
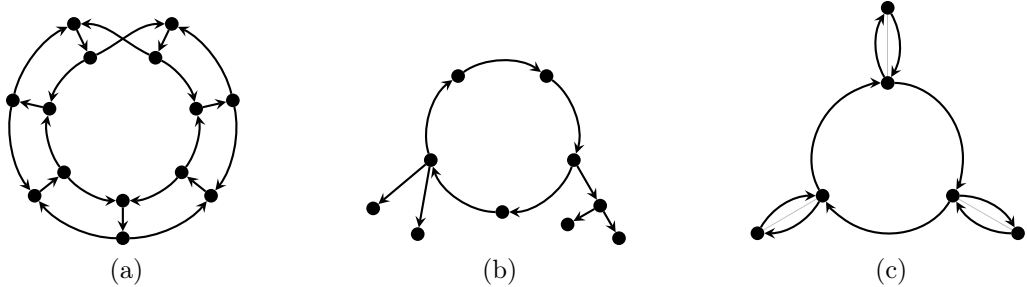
An example of a graphical game is the~$2$-cycle on two vertices
\raisebox{0.5ex}{\tikz[baseline]{%
    \fill (0,0) coordinate (a) circle (0.5ex);
    \fill (1,0) coordinate (b) circle (0.5ex);
    \draw[arc] (a) to[bend left] (b);
    \draw[arc] (b) to[bend left] (a);
}
}.
This game amounts to a test of all three constraints (for two parties, the notions of static causal order, definite causal order, and bi-causal order coïncide).
If the two parties are arranged along a partial order, then they will \emph{fail} in at least one of two cases,
i.e., whenever the selected arc disagrees with the partial order.
In that case, the receiving party correctly guesses the bit with half probability.
Overall, the chance of winning this game is upper bounded by~$(1+1/2)/2=3/4$.
A~violation of this inequality proves the incompatibility of the correlations with the constraint of a (static/definite/bi-)causal order.

Generally, it is difficult~\cite{pitowsky1991} to decide whether a given distribution~$P_{\underline A|\underline X}$ belongs to a polytope.
A naïve approach is to enumerate all \emph{facet-defining inequalities,}
and to verify that~$P_{\underline A|\underline X}$ satisfies each one of them.
In our case, however, complete enumerations are unknown.
Moreover, numerical computations suggest that in the case of \emph{static causal order} and \emph{bi-causal order,} novel families of inequalities emerge with an increasing number of parties;
the more parties are considered, the more nontrivial limitations are encountered.
In contrast, the inequalities that define the set of correlations compatible with \emph{definite causal order} seem to follow a simple pattern.
This motivates us to define \emph{weakly causal correlations} as the set of correlations that satisfy all these inequalities.
We then present an algorithm that runs in polynomial time in the length of the input (in fact, even in number of parties) to decide membership in that polytope:
Deciding whether some correlations are weakly causal is in~$\cc{P}$.

\subsection{Outlook}
In the next section we present some basic notation, and the required preliminaries on polytope theory and graph theory.
After that, we introduce the \emph{single-output scenario} and establish a map from polytopes to sets of digraphs.
This is followed by the definition of graphical tests, and the study of each of the causal constraints.
After that, we define \emph{weakly causal correlations} and present the efficient algorithm for deciding membership.
Finally, we conclude and discuss some open questions.

\section{Preliminaries}
The number of parties is always denoted by~$n$, and it is at least~$2$.
Conditional probability distributions are expressed analogous to~$P_{A,B|X,Y}$, and its probabilities are~$p(a,b|x,y)$.
Often we refer to the distribution via its probabilities only.
We use~$[\ell]$ for the set~$\{0,1,\dots,\ell-1\}$, and~$[\ell]^2_{\neq}$ for the set of distinct pairs over~$[\ell]$.
The infix operator~$\oplus$ denotes addition modulo~$2$.
On~$0/1$ vectors, the operator acts element-wise.
We use bold letters for vectors.
Their dimension is often implicit.
The entries of doubly-indexed vectors, e.g.,~$\bm v=(\bm v_{i,j})_{i,j}$, are ordered lexicographically.
These entries may also be addressed using a single index.
The vector~$\bm 0$ is the all-zero vector, and the vector~$\bm 1_i$ is~$\bm 0$ with a~$1$ at coordinate~$i$.
Linear inequalities are pairs~$(\bm v, c)$ and state~$\bm v\cdot\bm c =\sum_i \bm v_i \bm x_i \leq c$ for any~$\bm x$.
A linear inequality~$(\bm v, c)$ is \emph{non-negative,} if all entries of~$\bm v$ are non-negative,
and it is \emph{nontrivial,} if two or more entries are nonzero.
Usually, we use calligraphic letters for sets.
If some function $f$ is defined on the elements of a set~$\mathcal S$,
then~$f(\mathcal S)$ is~$\{f(x)|x\in\mathcal S\}$.
The \emph{cardinality} of a set~$\mathcal S$ is~$|\mathcal S|$.

\subsection{Polytope Theory}
All polytopes~$\mathcal P\subseteq \mathbb R^d$ we consider are convex.
In its V-representation, a~polytope is the convex hull of a finite set of vectors~$\mathcal S\subseteq \mathbb R^d$ ($\mathcal P = \conv(\mathcal S))$.
Equivalently, in its H-representation, a polytope is the intersection of finitely many halfspaces ($\mathcal P=\{\bm x \mid \bm A \bm x \leq \bm z\}$).
The set of \emph{extremal points} of~$\mathcal P$ is~$\ext(\mathcal P):=\{\bm p \in\mathcal P\mid \bm p\not\in\conv(\mathcal P\setminus\{\bm p\})\}$.
An inequality~$(\bm v, c)$ is \emph{valid for~$\mathcal P$,} if it satisfied by all elements of~$\mathcal P$.
The \emph{dimension~$\dim(\mathcal P)$ of~$\mathcal P$} is the dimension of its affine hull, and~$\mathcal P$ is full-dimensional if~$\dim(\mathcal P)=d$.
For a~full-dimensional polytope~$\mathcal P$, an inequality~$(\bm v,c)$ is \emph{facet-defining} whenever it is valid for~$\mathcal P$ with~$\dim(\mathcal P \cap \{\bm p \mid \bm v\cdot\bm p=c\})=d-1$.

\subsection{Graph Theory}
Throughout this article, all graphs are directed and simple.
A digraph~$D$ is defined as the tuple~$(\mathcal V(D), \mathcal A(D))$, where~$\mathcal V(D)$ is the set of \emph{vertices,}
and~$\mathcal A(D)\subseteq \{ i\arc j | i,j\in\mathcal V(D), i\neq j\}$ is the set of \emph{arcs.}
The \emph{order} of a~digraph~$D$ is~$|\mathcal V(D)|$.
A \emph{path} is a sequence of arcs~$(v_0\arc v1, v1\arc v2,\dots)$ where vertices are not revisited.
A \emph{cycle} is a path where the first and the last vertex coïncide.
A~$k$-cycle is a cycle with~$k$ arcs.
A digraph is \emph{acyclic} if it does not contain any cycle.
The set~$\DAG_n$ is the set of all acyclic digraphs~$D$ with~$\mathcal V(D)=[n]$.
The \emph{in-degree}~$\deg_D^\text{in}(j)$ of~$j\in\mathcal V(D)$ in~$D$ is~$|\{i\mid i\arc j\in\mathcal A(D)\}|$.
A \emph{source} is a vertex with zero in-degree.
For a~digraph~$D$ with~$\mathcal V(D)=[n]$, we use~$\bm\alpha(D)\in\{0,1\}^{n(n-1)}$ for the adjacency vector with~$\bm\alpha(D)_{i,j}=1\Leftrightarrow i\arc j\in\mathcal A(D)$.
Two digraphs~$D,E$ are \emph{isomorphic,} denoted by~$D\cong E$, if they are the same up to a relabeling of the vertices.
For a digraph~$D$ and a vertex~$i\in\mathcal V(D)$, we use~$\suc_D(i)$ for the set of all successors of~$i$ in~$D$, i.e., all vertices~$j$ for which there exists a path from~$i$ to~$j$.
A digraph~$D$ is \emph{strong} if for all~$i\in\mathcal V(D)$ we have~$\suc_D(i)=\mathcal V(D)$.
A \emph{strongly connected component of~$D$} is a set~$\mathcal C\subseteq\mathcal V(D)$ with~$\forall i \in\mathcal C:\suc_D(i)=\mathcal C$ that cannot be enlarged.
The vertices~$\mathcal V(D)$ can always be partitioned into strongly connected components~$\{\mathcal C_i\}_i$.
The digraph over all strongly connected components with an arc~$\mathcal C_i \arc \mathcal C_j$ whenever there is some~$u\in\mathcal C_i$ and~$v\in\mathcal C_j$ with~$u\arc v\in\mathcal A(D)$, is the \emph{condensation of~$D$.}
The condensation is acyclic.

\section{A Simple Scenario}
Bell tests~\cite{brunner2014} always involve the outputs of \emph{two or more} parties.
In the class of XOR games~\cite{cleve2004}, for instance, the objective of the~$n$ parties is to produce Boolean outputs such that their parity~$\bigoplus_{i\in[n]}A_i$ equals some function of their inputs.
In fact, Bell games become trivial if all outputs but of one party are ignored.
Here, in contrast, it turns out that considering a \emph{single} output is sufficient to make nontrivial statements.
The \emph{single-output} scenario we consider is the following (see also Fig.~\ref{fig:scenario}).
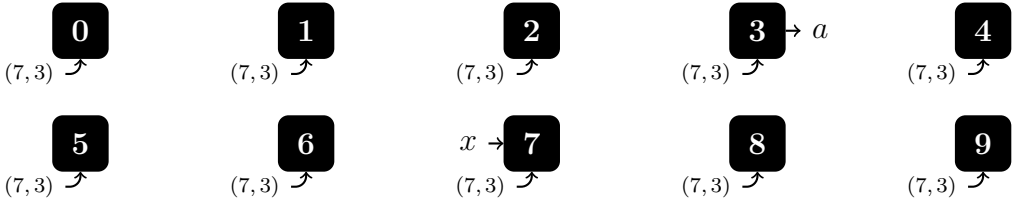
\begin{figure}
  \def\s{7}
  \def\r{3}
  \centering
  \begin{tikzpicture}
    \def\alen{0.2}
    \def\partywidth{0.75}
    \pgfmathsetmacro{\partyxspace}{3*\partywidth}
    \pgfmathsetmacro{\partyyspace}{1*\partywidth}
    \def\tot{10}
    \def\ncols{5}
    \foreach \i in {1, 2, ..., \tot} {
      \pgfmathsetmacro{\label}{int(\i-1)}
      \pgfmathsetmacro{\col}{mod(\label,\ncols)}
      \pgfmathsetmacro{\row}{int(\label/\ncols)}
      \pgfmathsetmacro{\xcoordinate}{\col*(\partywidth+\partyxspace)}
      \pgfmathsetmacro{\ycoordinate}{-\row*(\partywidth+\partyyspace)}
      \pgfmathsetmacro{\xcenter}{\xcoordinate+\partywidth/2}
      \pgfmathsetmacro{\ycenter}{\ycoordinate+\partywidth/2}
      \fill[rounded corners] (\xcoordinate,\ycoordinate) rectangle ++(\partywidth,\partywidth) node[pos=.5] {\color{white}{$\bm \label$}};
      \ifnum\label=\s
        \draw[thick,<-] (\xcoordinate,\ycenter) -- ++(-\alen,0) node[left] {$x$};
      \fi
      \ifnum\label=\r
        \draw[thick,->] (\xcoordinate+\partywidth,\ycenter) -- ++(\alen,0) node[right] {$a$};
      \fi
      \draw[thick,<-,rounded corners] (\xcenter,\ycoordinate) |- ++(-\alen,-\alen) node[left] {\scriptsize$(\s,\r)$};
    }
  \end{tikzpicture}
  \caption{%
    Illustration of the single-output scenario with~$n=10$.
    Each box represents a~party.
    The particular arrangement of the parties in this figure is arbitrary.
    Each party receives as an input the labels of the \emph{sender} and \emph{receiver,} here~$(s=\s,r=\r)$.
    The sender~$s$ has the additional input~$x$, and the receiver is the only party who produces an output ($a$).
    Different~causal~constraints~limit the set of attainable correlations~$P_{A|S,R,X}$~in~different~ways.
  }
  \label{fig:scenario}
\end{figure}
\begin{definition}[Single-output scenario]
  The~$n$ parties involved in the scenario are~$[n]$.
  The party that produces an output is~$r\in[n]$, called \emph{receiver,}
  and the output is a bit~$a\in\{0,1\}$.
  Moreover, a single party~$s\in[n]$ different from the receiver has a Boolean input~$x\in\{0,1\}$.
  We call this party \emph{the sender.}
  Finally, all parties~$[n]$ share the tuple~$(s,r)\in [n]^2_{\neq}$ as input.
  Correlations in this \emph{single-output scenario} are described by conditional probability distributions of the form~$P_{A|S,R,X}$.
\end{definition}
This definition does not specify in any way how the parties are connected or what signals may be transmitted.
These causal constraints will be enforced on top of this scenario in Section~\ref{sec:graphicalgames}.
But before doing so, we discuss the representation of correlations in this this scenario, and some general properties.

\subsection{Geometric Representation}
For~$n$ parties, the conditional probability distribution~$P_{A|S,R,X}$ is fully specified by the~$2n(n-1)$ probabilities~$p(0|s,r,x)$ for the receiver producing output~$0$.
Therefore, a specific~$P_{A|S,R,X}$ is represented by the probability vector~$\bm p:=(\bm p^0, \bm p^1)\in\mathbb R^{2d}$, with~$\bm p^x:=(p(0|s,r,x))_{s,r}\in\mathbb R^{d}$, and~$d:=n(n-1)$.
If we request that~$P_{A|S,R,X}$ must satisfy some causal constraint, then we focus on the body~$\mathcal P\subseteq\mathbb R^{2d}$ of such vectors in agreement with that constraint.
The dimension of the body in this scenario scales \emph{at most quadratically} with the number of parties.
This is in radical contrast to the usual Bell scenarios, where the dimensions grow exponentially in~$n$.

\subsection{Basic Properties}
As we show below, all causal constraints of interest produce bodies~$\mathcal P\subseteq\mathbb R^{2d}$ that satisfy some useful properties.
\begin{definition}[Operational 0/1 polytope,~$\Delta_n$]
  The body~$\mathcal P\subseteq \mathbb R^{2d}$ of single-output correlations is an \emph{operational 0/1 polytope} if and only if
  it is a convex polytope with~$\ext(\mathcal P)\subseteq\{0,1\}^{2d}$ and
  \begin{align}
    (\bm p^0, \bm p^1) \in \ext(\mathcal P)
    \Longrightarrow
    \forall (s, r)\in[n]^2_{\neq}:
    (\bm p^0 \oplus \bm 1_{(s, r)},
    \bm p^1 \oplus \bm 1_{(s, r)})
    \in \ext(\mathcal P)
    \,.
    \label{eq:opclosed}
  \end{align}
  The set of all such polytopes for~$n$ parties is~$\Delta_n$.
\end{definition}
For some~$\mathcal P$ to be \emph{a~$0/1$ polytope} means that the set of single-output correlations is closed under convex combinations,
and that the extremal points correspond to deterministic behaviors.
The last property (Eq.~\eqref{eq:opclosed}) states the \emph{operational closure} of the set.
In order to see that, focus on some specific sender~$\hat s$ and receiver~$\hat r$, and consider the entries
\begin{equation}
  p(0 | \hat s, \hat r, 0) = c_0
  \,,
  \quad
  p(0 | \hat s, \hat r, 1) = c_1
  \,,
\end{equation}
with~$c_0,c_1 \in \{0,1\}$.
A first observation is that whenever~$c_0 \neq c_1$, then the output of~$\hat r$ depends on the input of~$\hat s$: the sender~$\hat s$ signals to the receiver~$\hat r$.
In particular, the case~$c_0=1$ and~$c_1=0$ describes the situation where the receiver~$\hat r$ deterministically outputs~$x$ whenever~$\hat s$ is the sender.
By flipping these entries we end in the situation where~$\hat r$ deterministically outputs~$x\oplus 1$ whenever~$\hat s$ is the sender.
The remaining two cases, i.e., when~$c_0=c_1=0$ and when~$c_0=c_1=1$, describe situations where~$\hat r$ outputs a constant~$1$ ($0$) whenever~$\hat s$ is the sender.
Now, suppose some correlations~$P_{A|S,R,X}$ satisfy a~causal constraint (which we did not specify yet),
then, the same correlations where party~$\hat r$ flips the output whenever the sender is~$\hat s$ also satisfy the same constraint.
After all, this amounts to a basic relabeling of party~$\hat r$'s output.

\subsection{Signaling Relations}
The operational closure in the above definition induces an equivalence relations on the extremal points~$\ext(\mathcal P)\in\{0,1\}^{2d}$:
the vectors~$\bm p$ and~$\bm q$ are equivalent if and only if~$\bm q$ is reachable from~$\bm p$ through the repetitive application of the implication~\eqref{eq:opclosed}.
This equivalence relation is compactly defined as
\begin{equation}
  \bm p \sim \bm q
  :\Longleftrightarrow
  \left(\bm p^0 \oplus \bm p^1\right)
  =
  \left(\bm q^0 \oplus \bm q^1\right)
  \,.
\end{equation}
Since the vector~$\bm p^0 \oplus \bm p^1$ contains a Boolean entry for every~$(i,j)\in[n]^2_{\neq}$,
each equivalence class~$[\bm p]$ itself is a \emph{digraph:}
\begin{equation}
  \mathcal V([\bm p])
=
[n]
\,,
\quad
\mathcal A([\bm p])
=
\{i \arc j \mid (\bm p^0 \oplus \bm p^1)_{i,j} = 1\}
\,.
\end{equation}
This digraph has a natural interpretation.
The presence of the arc~$i\arc j$ states that the probabilities~$p(0|i,j,0)$ and~$p(0|i,j,1)$ \emph{differ.}
Conversely, the absence of the arc~$i\arc j$ states~$p(0|i,j,0)=p(0|i,j,1)$.
So, by going back to the intuition built above,~$[\bm p]$ has an arc~$i\arc j$ if and only if the correlations~$P_{A|S,R,X}$ are signaling from~$i$ to~$j$.
This, in conclusion, provides us with a map~$\Gamma$ from \emph{operational 0/1 polytopes} of single-output correlations
to \emph{sets of digraphs} that encode signaling relations:
\begin{alignat}{3}
  \label{eq:gamma}
  \Gamma
  &: \Delta_n &\rightarrow 2^{\mathcal G_n}
  &&\ :: \mathcal P &\mapsto \{[\bm p] \mid \bm p\in\ext(\mathcal P)\}
  \,,
\end{alignat}
where~$2^{\mathcal G_n}$ is the powerset of the set of digraphs over~$[n]$.

\subsection{Polytope of Digraphs}
The map~$\Gamma$ provides the means of studying operational 0/1 polytopes via studying sets of digraphs.
Just as before, where we geometrically represent a~set of single-output correlations as a polytope~$\mathcal P\in\mathbb R^{2d}$,
we now geometrically represent a set~$\mathcal D$ of digraphs as a polytope~$\mathcal Q:=\conv(\bm\alpha(\mathcal D))\in\mathbb R^d$.
As we show in Ref.~\cite{tselentis2023}, in certain cases, the facet-defining inequalities of~$\mathcal Q$ can be lifted to represent facet-defining inequalities of~$\mathcal P$.
With this, this detour to sets of digraphs helps us to derive the facet-defining inequalities of~$\mathcal P$ by solving the same problem \emph{on half of the original dimension.}
\begin{lemma}[Lifting~\cite{tselentis2023}]
  \label{lemma:lifting}
  If an operational 0/1 polytope~$\mathcal P$ and its digraph polytope~$\mathcal Q:=\conv(\bm\alpha(\Gamma(\mathcal P)))$ are full-dimensional,
  and if~$(\bm w, c)$ is a non-negative and nontrivial facet-defining inequality of~$\mathcal Q$,
  then~$((\bm w, -\bm w), c)$ is a~facet-defining inequality of~$\mathcal P$.
\end{lemma}
For completeness and as a warm-up, we restate the proof here.
\begin{proof}
  The validity of~$((\bm w, -\bm w), c)$ for~$\mathcal P$ follows directly from the validity of~$(\bm w, c)$ for~$\mathcal Q$, i.e., for some~$\bm p\in\ext(\mathcal P)$,
  we have~$(\bm p^0 \oplus \bm p^1)\in\ext(\mathcal Q)$, and
  \begin{equation}
    (\bm w, -\bm w)\cdot \left(\bm p^0, \bm p^1\right)
    =
    \bm w\cdot\left(\bm p^0 - \bm p^1\right)
    \leq
    \bm w\cdot\left(\bm p^0 \oplus \bm p^1\right)
    \leq
    c
    \,.
  \end{equation}

  In order to show that~$((\bm w, -\bm w),c)$ is facet-defining for~$\mathcal P$, we generate a set of~$2d$ affinely independent extremal points of~$\mathcal P$ that saturate the inequality.
  Since~$(\bm w, c)$ is facet-defining for~$\mathcal Q$, there exists a set~$\mathcal T\subseteq\ext(\mathcal Q)$ of~$d$ affinely independent vectors that saturate the inequality, i.e., with
  \begin{equation}
    \forall \bm q \in\mathcal T:\bm w \cdot \bm q = c
    \,.
  \end{equation}
  These correspond to~$d$ affinely independent vectors~$\mathcal S_0 := \{(\bm q, \bm 0) \mid \bm q \in\mathcal T\}\subseteq\mathcal P$
  that saturate the \emph{lifted inequality:}
  \begin{equation}
    \forall \bm p \in \mathcal S_0:
    (\bm w, -\bm w)\cdot \bm p
    =
    c
    \,.
  \end{equation}
  The remaining~$d$ extremal points are obtained from taking elements in~$\mathcal S_0$ and by flipping certain bits, producing a \emph{diagonal} in the last~$d$ dimensions.
  First note that, if for some~$\bm p=(\bm q,\bm 0)\in\mathcal S_0$, the vector~$\bm q$ has a~$0$ entry at position~$i\in[d]$, then
  \begin{equation}
    (\bm w, -\bm w)
    \cdot
    (\bm p \oplus(\bm 1_{i}, \bm 1_{i}))
    =
    (\bm w, -\bm w)
    \cdot
    (\bm q \oplus \bm 1_i, \bm 1_i)
    =
    \bm w \cdot \bm q
    +
    w_i
    -
    w_i
    = c
    \,.
  \end{equation}
  The remaining~$d$ extremal points that saturate the lifted inequality are thus
  \begin{equation}
    \mathcal S_1
    :=
    \{
      (
      \bm q^{(i)}
      \oplus \bm 1_i,
      \bm 1_i)
      \mid
      i\in[d]
    \}
    \,,
  \end{equation}
  where~$\bm q^{(i)}\in\mathcal T$ has a~$0$ entry at position~$i$.
  The set~$\mathcal S_0 \cup \mathcal S_1$ describes~$2d$ affinely independent vectors that saturate the lifted inequality,
  unless there exists some~$i_0\in[d]$ for which we cannot find any~$\bm q\in\mathcal T$ with a~$0$ entry at position~$i_0$.
  But this latter situation cannot arise.
  Firstly, if there is a second index~$j_0\in[d]$ with~$\forall \bm q\in\mathcal T: \bm q_{j_0}=1$,
  then~$\mathcal T$ (which defines a~$(d-1)$-dimensional space) does not fit into the remaining~$d-2$ coordinates.
  Lastly, if there is no such second index, then~$\mathcal T$ occupies the whole remaining space.
  This, by assuming without loss of generality~$i_0=0$, means that
  \begin{equation}
    \forall \ell\in[d-1]: \bm w\cdot(1,\bm 1_\ell) = c
    \,,
  \end{equation}
  which in turn implies~$\bm w=(c,\bm 0)$, i.e., the inequality is trivial.
\end{proof}

In fact, the symmetry induced by the operational closure allows us to \emph{generate an exponential number} of facet-defining inequalities.
We will need this lemma for the second part of the present article.
\begin{lemma}[Facet generation]
  \label{lemma:generative}
  If~$((\bm w, -\bm w), c)$ is a facet-defining inequality of a~full-dimensional operational 0/1 polytope~$\mathcal P$,
  then~$((\bm\phi \bm w, -\bm\phi \bm w),c)$ is a~facet-defining inequality of~$\mathcal P$ for any~\mbox{$\bm\phi\in\{-1,+1\}^d$}, where the product is taken element-wise.
\end{lemma}
\begin{proof}
  It is sufficient to consider~$\bm\phi$ with a single~$-1$ at coordinate~$i_0\in[d]$.
  The full statement follows from the iterative application of this case.
  This case holds because flipping the signs according to~$\bm \phi$ is tantamount
  to flipping the corresponding entries in the extremal point~$\bm p=(\bm p^0,\bm p^1)$:
  \begin{align}
    (\bm \phi \bm w,-\bm \phi \bm w) \cdot \bm p
    &=
    \bm\phi\bm w\cdot\left(\bm p^0-\bm p^1\right)
    \\
    &=
    \sum_{i\neq i_0}
    \bm w_i
    \left(
      \bm p^0_i
      -
      \bm p^1_i
    \right)
    -
    \bm w_{i_0}
    \left(
      \bm p^0_{i_0}
      -
      \bm p^1_{i_0}
    \right)
    \\
    &=
    \sum_{i\neq i_0}
    \bm w_i
    \left(
      \bm p^0_i
      -
      \bm p^1_i
    \right)
    +
    \bm w_{i_0}
    \left(
      \bm p^1_{i_0}
      -
      \bm p^0_{i_0}
    \right)
    \\
    &=
    (\bm w,-\bm w)\cdot
    \left(
      \bm p^0\oplus \bm 1_{i_0}
      ,
      \bm p^1\oplus \bm 1_{i_0}
    \right)
    \,.
  \end{align}
  Thus, the ``rotated'' inequality~$((\bm \phi\bm w,-\bm\phi\bm w),c)$ is valid for~$\mathcal P$, and there are~$2d$ affinely independent vectors in~$\mathcal P$ that saturate it.
\end{proof}

\section{Graphical Games}
\label{sec:graphicalgames}
Before we consider specific causal constraints, we discuss graphical games in general.
Take some digraph polytope~$\mathcal Q\subseteq \mathbb R^d$, and suppose that~$(\bm w, c)$ is an inequality valid for~$\mathcal Q$.
This inequality states that
\begin{equation}
  \bm w \cdot \bm r
  =
  \sum_{(i,j)\in[n]^2_{\neq}}
  w_{i,j}
  r_{i,j}
  \leq
  c
\end{equation}
for all vectors~$\bm r\in\mathcal Q$.
The form of these vectors allow for a more graph-centric reading of the inequality:
An entry at coordinate~$(i,j)$ specifies a weight of the arc~$i\arc j$ over the vertices~$[n]$.
In particular, when we have~$\bm w\in\{0,1\}^d$, then the inequality is
\begin{equation}
  \sum_{i\arc j\in\mathcal A(\bm w)}
  r_{i,j}
  \leq
  c
  \,;
\end{equation}
the sum of the entries in~$\bm r$ that correspond to arcs in~$\bm w$ is upper bounded by~$c$ (see Fig.~\ref{fig:gg}).
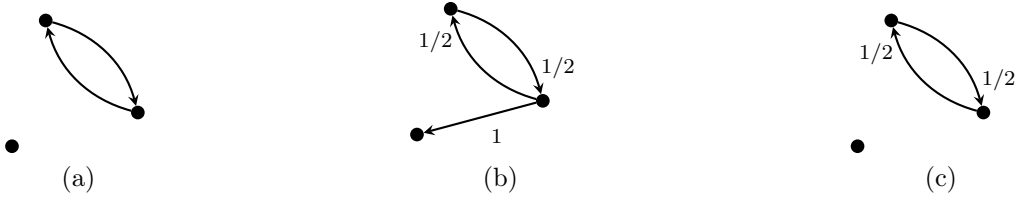
\begin{figure}
  \centering
  \def\adelta{45}
  \subfloat[\label{subfig:gtex}]{%
    \begin{tikzpicture}
      \pgfmathsetmacro{\angledelta}{360/3}
      \foreach \x in {1, 2, 3} {
        \pgfmathsetmacro{\angle}{(\x-1) * \angledelta + \angledelta/2 + \adelta}
        \fill (\angle:\radiusmain) circle (\vertexradius);
      }
      \pgfmathsetmacro{\anglea}{(1-1) * \angledelta + \angledelta/2 + \adelta}
      \pgfmathsetmacro{\angleb}{(3-1) * \angledelta + \angledelta/2 + \adelta}
      \draw[arc] (\anglea:\radiusmain) to[bend left] (\angleb:\radiusmain);
      \draw[arc] (\angleb:\radiusmain) to[bend left] (\anglea:\radiusmain);
    \end{tikzpicture}
  }
  \hfill
  \subfloat[\label{subfig:targetex}]{%
    \begin{tikzpicture}
      \pgfmathsetmacro{\angledelta}{360/3}
      \foreach \x in {1, 2, 3} {
        \pgfmathsetmacro{\angle}{(\x-1) * \angledelta + \angledelta/2 + \adelta}
        \fill (\angle:\radiusmain) circle (\vertexradius);
      }
      \pgfmathsetmacro{\anglea}{(1-1) * \angledelta + \angledelta/2 + \adelta}
      \pgfmathsetmacro{\angleb}{(3-1) * \angledelta + \angledelta/2 + \adelta}
      \pgfmathsetmacro{\anglec}{(2-1) * \angledelta + \angledelta/2 + \adelta}
      \draw[arc] (\anglea:\radiusmain) to[bend left] node[right,near end] {\scriptsize $1/2$} (\angleb:\radiusmain);
      \draw[arc] (\angleb:\radiusmain) to[bend left] node[left,near end] {\scriptsize $1/2$} (\anglea:\radiusmain);
      \draw[arc] (\angleb:\radiusmain) -- node[below right,midway] {\scriptsize $1$} (\anglec:\radiusmain);
    \end{tikzpicture}
  }
  \hfill
  \subfloat[\label{subfig:prodex}]{%
    \begin{tikzpicture}
      \pgfmathsetmacro{\angledelta}{360/3}
      \foreach \x in {1, 2, 3} {
        \pgfmathsetmacro{\angle}{(\x-1) * \angledelta + \angledelta/2 + \adelta}
        \fill (\angle:\radiusmain) circle (\vertexradius);
      }
      \pgfmathsetmacro{\anglea}{(1-1) * \angledelta + \angledelta/2 + \adelta}
      \pgfmathsetmacro{\angleb}{(3-1) * \angledelta + \angledelta/2 + \adelta}
      \pgfmathsetmacro{\anglec}{(2-1) * \angledelta + \angledelta/2 + 90}
      \draw[arc] (\anglea:\radiusmain) to[bend left] node[right,near end] {\scriptsize $1/2$} (\angleb:\radiusmain);
      \draw[arc] (\angleb:\radiusmain) to[bend left] node[left,near end] {\scriptsize $1/2$} (\anglea:\radiusmain);
    \end{tikzpicture}
  }
  \caption{%
    (a) Example of a graphical game~$\bm w$ as a $2$-cycle on three vertices.
    Here, and in the other digraphs, we omit the vertex labels.
    (b) A target digraph~$\bm r$ with weights.
    (c)~The ``inner product'' amounts to the value~$1$.
  }
  \label{fig:gg}
\end{figure}

According to Lemma~\ref{lemma:lifting}, the \emph{lifted inequality} of~$(\bm w, c)$ is~$((\bm w, -\bm w), c)$.
So, evaluating this inequality on a conditional probability distribution~$P_{A|S,R,X}$ amounts to evaluating the inequality
\begin{equation}
  (\bm w, -\bm w)\cdot(\bm p^0, \bm p^1) \leq c
  \,,
\end{equation}
where we use the geometric representation~$\bm p=(\bm p^0,\bm p^1)$ of~$P_{A|S,R,X}$.
By returning to the probability notation, and by taking into account that~$\bm w$ itself is a digraph, the left-hand side of this inequality is
\begin{align}
  \sum_{i\arc j\in\mathcal A(\bm w)}
  &
  p(0 | i,j,0)
  -
  p(0 | i,j,1)
  \\
  &=
  \sum_{i\arc j\in\mathcal A(\bm w)}
  p(0 | i,j,0)
  -
  (1-p(1 | i,j,1))
  \\
  &=
  \sum_{\substack{i\arc j\in\mathcal A(\bm w)\\x\in\{0,1\}}}
  p(x | i,j,x)
  -
  |\mathcal A(\bm w)|
  \,.
\end{align}
This allows us to rewrite the inequality conveniently as
\begin{equation}
  \Pr[A=X] = \frac{1}{2|\mathcal A(\bm w)|}
  \sum_{\substack{i\arc j\in\mathcal A(\bm w)\\x\in\{0,1\}}}
  p(x | i,j,x)
  \leq
  \frac{1}{2}
  +
  \frac{c}{2|\mathcal A(\bm w)|}
  \,.
\end{equation}
This inequality bounds the probability that the receiver outputs the uniformly distributed bit given to the sender,
where the sender-receiver pair is chosen uniformly at random from the digraph~$\bm w$.

With this, we are in position to define \emph{graphical games} and \emph{graphical tests:}
\begin{definition}[Graphical game, graphical test]
  \label{def:gggt}
  A \emph{graphical game} for~$n$ parties~$[n]$ is a digraph~$G$ over~$\mathcal V(G)=[n]$.
  The game is played as follows.
  A~referee picks uniformly at random an arc~$s\arc r$ from~$\mathcal A(G)$ and a bit~$x$.
  The referee announces the arc~$s\arc r$ to all parties, and announces~$x$ to the sender~$s$.
  The parties \emph{win the game~$G$,} which we express by~$\win(G)$, whenever the receiver~$r$ outputs~$x$.
  A \emph{graphical test}~$\graphicaltest(G,t)$ is a graphical game~$G$ with a bound~$t$.
  The parties \emph{pass the test~$\graphicaltest(G,t)$} whenever they saturate the inequality~$\Pr[\win(G)] \leq t$.
  Otherwise, they \emph{fail the test.}
\end{definition}

In what follows, we describe the graphical games with which certain causal constraints can be tested.
The recipe is always the same:
\begin{enumerate}
  \item Show that the correlations of interest form an operational 0/1 polytope~$\mathcal P$.
  \item Derive the corresponding polytope of digraphs~$\mathcal Q$ using the map~$\Gamma$.
  \item Compute the facet-defining inequalities of~$\mathcal Q$, and identify the graphical games.
\end{enumerate}
In all three cases, the first step is immediate.
It follows directly from the definition and the discussion on operational 0/1 polytopes above.
Notably, the \emph{style} in which the digraphs are defined in each of the three cases differs.
For static causal order, we make use the \emph{graph isomorphisms} to define the games,
for definite causal order, then again, we \emph{parametrically} define all identified games,
and finally, for bi-causal order, we abstractly define the relevant \emph{set of digraphs.}
We believe that this mirrors the computational difficulty in deciding membership in the respective cases.

\subsection{Static Causal Order}
The causal constraint of \emph{static causal order} requires that all communication follows a partial order.
This corresponds to \emph{non-adaptive} strategies; the flow of information is predefined (although, it may be probabilistic).
\begin{definition}[Static causal order,~$\mathcal C_n^\text{static}$]
  \label{def:static}
  The~$n$-party correlations~$P_{A|S,R,X}$ agree with \emph{static causal order} if and only if
  their probabilities decompose as
  \begin{equation}
    p(a|s,r,x)
    =
    \sum_{\sigma: s \preceq_\sigma r}
    p(\sigma)
    p^{\preceq}_\sigma(a|s,r,x)
    +
    \sum_{\sigma: s \not\preceq_\sigma r}
    p(\sigma)
    p^{\not\preceq}_\sigma(a|s,r)
    \,,
  \end{equation}
  where~$\preceq_\sigma$ is a partial order over~$[n]$.
  The set of such correlations is~$\mathcal C_n^\text{static}$.
\end{definition}

\subsubsection{DAG Polytope}
For any~$n$, the set~$\mathcal C_n^\text{static}$ forms an operational 0/1 polytope~$\mathcal P_n^\text{static}$.
As shown in Ref.~\cite{tselentis2023}, this polytope is characterized by directed acyclic graphs:
\begin{theorem}[DAGs~\cite{tselentis2023}]
  The set of digraphs~$\Gamma(\mathcal P_n^\text{static})$ is the set of all directed acyclic graphs on~$n$ vertices: $\Gamma(\mathcal P_n^\text{static})  = \DAG_n$.
\end{theorem}
The proof (see, Ref.~\cite{tselentis2023}) is straightforward:
The~$\subseteq$-inclusion follows from the underlying partial order of communication,
and the~$\supseteq$-inclusion follows from the fact that the transitive closure of any directed acyclic graph is a~partial order.

As by our recipe, we now study the polytope of directed acyclic graphs.
\begin{definition}[DAG polytope,~$\mathcal Q_n^\text{DAG}$]
  The \emph{polytope of directed acyclic graphs} (DAG polytope, for short) is~$\mathcal Q_n^\text{DAG} := \conv(\bm \alpha (\DAG_n))$.
\end{definition}

This polytope of directed acyclic graphs has been extensively studied by Grötschel, Jünger, and Reinelt~\cite{grotschel1985} in the mid 80s.
They studied these graphs in the context of the feedback-arc-set problem --- a member of the famous list of $\cc{NP-complete}$ problems by Karp~\cite{karp1972}.
This problem asks whether it is possible to remove~$k$ or fewer arcs from a digraph, such that the resulting digraph is acyclic.
The polytope approach~\cite{grotschel1985} tackles this problem from a~combinatorial-optimization perspective.

\subsubsection{Cycle, Fence, and Möbius Games}
Grötschel, Jünger, and Reinelt~\cite{grotschel1985} identified various classes of facet-defining inequalities of that polytope.
Their result also suggests that for higher dimensions, novel families of inequalities may arise; a complete characterization seems intractable.
All inequalities correspond to \emph{graphical tests.}
Here, we present the classes of~\emph{cycle inequalities,} \emph{fence inequalities,} and \emph{Möbius inequalities.}
\begin{definition}[$k$-cycle,~$k$-fence,~$k$-Möbius digraphs]
  The~\emph{$k$-cycle digraph}~$C_k^{n}$ with~$n\geq k\geq 2$,
  the~\emph{$k$-fence digraph}~$F_k^{n}$ with~$n\geq 2k\geq 4$,
  and the~\emph{$k$-Möbius digraph}~$M_k^{n}$ with~$n\geq 2k\geq 6$ and~$k$ odd, are
  \begin{align}
    \mathcal V\left(C_k^{n}\right)
    &:=
    \mathcal V\left(F_k^{n}\right)
    :=
    \mathcal V\left(M_k^{n}\right)
    :=
    [n]
    \,,
    \\
    \mathcal A\left(C_k^{n}\right)
    &:=
    \{i\arc (i+1)_k\mid i\in[k]\}
    \,,
    \\
    \mathcal A\left(F_k^{n}\right)
    &:=
    \{ 2i \arc (2i+1) \mid i\in[k]\}\cup\{ (2i+1) \arc 2j\mid (i,j)\in[n]^2_{\neq}\}
    \,,
    \\
    \mathcal A\left(M_k^{n}\right)
    &:=
    \{ i \arc (i+1)_{2k} \mid i\in [2k] \} \,\cup
    \{ (2i+3)_{2k} \arc 2i \mid i\in[k]\}
    \,,
  \end{align}
  where we use~$(x)_m$ as the short-hand notation for the smallest non-negative number that is congruent to~$x \bmod m$.
\end{definition}
The~$k$-cycle digraph is simply a cycle of length~$k$ embedded in a digraph of order~$n$ (see Fig.~\ref{subfig:gtex} for~$C_2^3$).
The~$k$-fence digraph, then again, contains the bipartite subdigraph of order~$2k$ with ``downward posts'' and ``upward rails'' (see Fig.~\ref{subfig:kfence}).
Note that the~$2$-fence digraph is the same as the~$4$-cycle digraph of the same order.
\begin{figure}
  \centering
  \def\n{15}
  \def\kfence{6}
  \def\kmobius{5}
  \def\height{2}
  \def\wfact{2}
  \subfloat[\label{subfig:kfence}]{%
    \begin{tikzpicture}
      \foreach \i in {1,...,\kfence} {
        \pgfmathsetmacro{\ulab}{int(2*(\i-1))}
        \pgfmathsetmacro{\llab}{int((2*(\i-1)+1)}
        \pgfmathsetmacro{\xco}{\i*\wfact}
        \fill (\xco,\height) circle (\vertexradius) node[above right] {\scriptsize $\ulab$};
        \fill (\xco,0) circle (\vertexradius) node[below left] {\scriptsize $\llab$};
        \draw[arc] (\xco,\height) -- (\xco,0);
        \foreach \j in {1,...,\kfence} {
          \pgfmathsetmacro{\xxco}{\j*\wfact}
          \ifnum\i>\j
            \draw[arc] (\xco,0) -- (\xxco,\height);
            \else\ifnum\i<\j
              \draw[arc] (\xco,0) -- (\xxco,\height);
            \fi \fi
          }
        }
        \pgfmathsetmacro{\rem}{int(\n-2*\kfence)}
        \pgfmathsetmacro{\xco}{(\kfence+1/2)*\wfact}
        \pgfmathsetmacro{\hdelta}{\height/(\rem-1)}
        \foreach \i in {1,...,\rem} {
          \pgfmathsetmacro{\yco}{\height-(\i-1)*\hdelta}
          \pgfmathsetmacro{\lab}{int(2*\kfence+-1+\i)}
          \fill (\xco,\yco) circle (\vertexradius) node [right] {\scriptsize $\lab$};
        }
      \end{tikzpicture}
    }

  \subfloat[\label{subfig:kmobius}]{%
    \begin{tikzpicture}
      \pgfmathsetmacro{\tot}{2*\kmobius}
      \foreach \i in {1,...,\tot} {
        \pgfmathsetmacro{\lab}{int(\i-1)}
        \pgfmathsetmacro{\xco}{int((\i-1)/2)*\wfact}
        \pgfmathsetmacro{\istop}{int(mod(int((\i)/2)+1,2))}
        \pgfmathsetmacro{\odd}{int(mod(\i-1,2))}
        \pgfmathsetmacro{\yco}{\istop*\height}
        \fill (\xco,\yco) circle (\vertexradius) node[above right] {\scriptsize $\lab$};
        \ifnum\i<\tot
          \ifnum\odd=1
            \draw[arc] (\xco,\yco) -- ++(\wfact,0);
            \ifnum\i=2
              \pgfmathsetmacro{\xcoalt}{int(\kmobius-1)*\wfact}
              \draw[arc,->] (\xcoalt,\yco) -- ++(-\wfact,0);
            \else
              \draw[arc,->] (\xco,\yco) -- ++(-\wfact,0);
            \fi
          \else
            \pgfmathsetmacro{\hdelta}{(-1)^\istop*\height}
            \draw[arc] (\xco,\yco) -- ++(0,\hdelta);
          \fi
        \fi
      }
      \pgfmathsetmacro{\xdel}{(\kmobius-1/2)*\wfact}
      \pgfmathsetmacro{\xfin}{(\kmobius-1)*\wfact}
      \draw[arc,rounded corners] (0,0) -- ++(0, -\height*1/5) -- ++(\xdel,0) |- (\xfin,\height);
      \draw[arc,rounded corners] (\xfin,0) -- ++(0, -\height*2/5) -- ++(-\xdel,0) |- (0,\height);
      \pgfmathsetmacro{\rem}{int(\n-2*\kmobius)}
      \pgfmathsetmacro{\xco}{\kmobius*\wfact}
      \pgfmathsetmacro{\hdelta}{\height/(\rem-1)}
      \foreach \i in {1,...,\rem} {
        \pgfmathsetmacro{\yco}{\height-(\i-1)*\hdelta}
        \pgfmathsetmacro{\lab}{int(2*\kmobius+-1+\i)}
        \fill (\xco,\yco) circle (\vertexradius) node[right] {\scriptsize $\lab$};
      }
    \end{tikzpicture}
  }
  \caption{Two classes of digraphs that describe graphical games for~$n=\n$ parties:
    (a)~the~$\kfence$-fence digraph of order~$\n$, and
    (b)~the~$\kmobius$-Möbius digraph of order~$\n$.
  }
	\label{fig:kfence}
\end{figure}
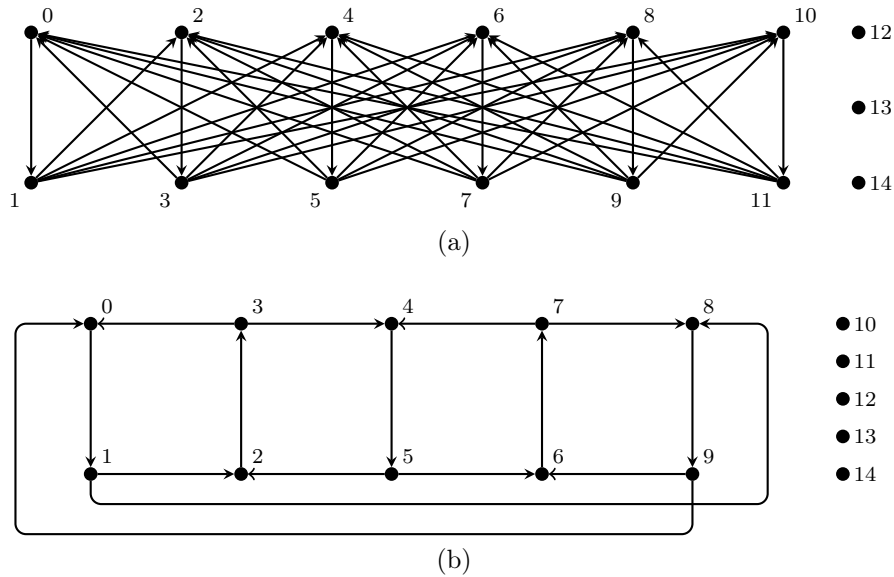
Finally, the~$k$-Möbius digraph is an orientation of the ladder graph of height~$k$, such that the internal faces form~$4$-cycles, and has the added arcs that cyclically closes the ladder (see Fig.~\ref{subfig:kmobius} and Fig.~\ref{subfig:moebius}).
The~$3$-Möbius digraph and the~$3$-fence digraph are the same.

\begin{theorem}[DAG facets~\cite{grotschel1985}]
  \label{thm:dagfacets}
  The following inequalities are facet-defining for the DAG polytope~$\mathcal Q_n^\text{DAG}$:
  \begin{align}
    & (- \bm 1_{i,j}, 0) & \text{(for any $(i,j)\in[n]^2_{\neq}$; trivial inequalities)}
    \,,
    \\
    & (\bm \alpha(G), k-1) & \text{(for any $G \cong C_k^n$; $k$-cycle inequalities)}
    \,,
    \\
    & (\bm \alpha(G), k^2-k+1) & \text{(for any $G \cong F_k^n$; $k$-fence inequalities)}
    \,,
    \\
    & (\bm \alpha(G), (5k-1)/2) & \text{(for any $G \cong M_k^n$; $k$-Möbius inequalities)}
    \,.
  \end{align}
\end{theorem}
Thanks to the lifting lemma (Lemma~\ref{lemma:lifting}), we thus obtain the following graphical tests:
\begin{corollary}[Static causal order, see Ref.~\cite{tselentis2023}]
  \label{cor:cyclefencemobius}
  The following inequalities describe~$n$-party graphical tests for static causal order:
  \begin{align}
    & \graphicaltest\left(G, 1-\frac{1}{2k}\right) & \text{(for any $G \cong C_k^n$; $k$-cycle test)}
    \,,
    \\
    & \graphicaltest\left(G, 1-\frac{k-1}{2k^2}\right) & \text{(for any $G \cong F_k^n$; $k$-fence test)}
    \,,
    \\
    & \graphicaltest\left(G, 1-\frac{k+1}{12k}\right) & \text{(for any $G \cong M_k^n$; $k$-Möbius test)}
    \,.
  \end{align}
  In each case, the graphical test is facet-defining for~$\mathcal P_n^\text{static}$.
\end{corollary}
Note that the bound for the~$k$-cycle test and for the~$k$-fence test approach the algebraic maximum for an increasing number~$k$.
In contrast, the bound of the~$k$-Möbius test never exceeds~$11/12$.
This makes the latter game favorable for tests involving a large number of parties.

\subsection{Definite Causal Order}
The assumption of \emph{definite causal order} captures the idea that the future does not influence the past.
This causal constraint relaxes the constraint of static causal order insofar that an event may influence \emph{everything} within its future.
Thus, an event may also specify the causal order among the events in its future.
Tailored to the scenario we consider, this means that the sender cannot signal to the receiver whenever the receiver is the \emph{first} in some causal order.
\begin{definition}[Definite causal order,~$\mathcal C_n^\text{causal}$]
  \label{def:causal}
  The~$n$-party correlations~$P_{A|S,R,X}$ agree with \emph{definite causal order} if and only if their probabilities decompose as
  \begin{equation}
    p(a|s,r,x)
    =
    p_\text{first}(r)
    p_0(a|s,r)
    +
    (1-p_\text{first}(r))
    p_1(a|s,r,x)
    \,,
  \end{equation}
  where~$p_\text{first}(r)$ is the probability for the receiver to be the first in the causal order.
  The set of such correlations is~$\mathcal C_n^\text{causal}$.
\end{definition}

\subsubsection{Source-Digraph Polytope}
The set~$\mathcal C_n^\text{causal}$ forms an operational 0/1 polytope~$\mathcal P_n^\text{causal}$.
Therefore, we are in position to use the map~$\Gamma$ and study the resulting set of digraphs instead.
\begin{theorem}[Source digraphs]
  The set of digraphs~$\Gamma(\mathcal P_n^\text{causal})$ is the set of all digraphs on~$n$ vertices with at least one source vertex:
  \begin{equation}
    \Gamma(\mathcal P_n^\text{causal})
    =
    \{ D \mid \mathcal V(D)=[n],\exists i \in[n]:\deg_D^\text{in}(i)=0 \}
    \,.
  \end{equation}
\end{theorem}
\begin{proof}
  For the~$\subseteq$-inclusion it is sufficient to note that since at least one party~$i_0$ is causally before all others, then that party cannot receive any signal:
  Party~$i_0$ describes a source.
  For the opposite inclusion, suppose that~$D$ is a digraph with a source~$i_0$.
  The following correlations~$P_{A|S,R,X}$ with the probabilities
  \begin{equation}
    p(0|s,r,x) =
    \begin{cases}
      1-x & (s,r)\in\mathcal A(D) \\
      1 & \text{otherwise,}
    \end{cases}
  \end{equation}
  respect that signaling structure (we have~$p(0|s,r,0)\oplus p(0|s,r,1)=1$ if and only if~$(s,r)\in\mathcal A(D)$)
  and are causal~(the probability~$p(0|s,i_0,x)$ is independent from~$x$).
\end{proof}
\begin{definition}[Source-digraph polytope, $\mathcal Q_n^\text{src}$]
  The \emph{polytope of source digraphs} is~$\mathcal Q_n^\text{src} := \conv(\Gamma(\mathcal P_n^\text{causal}))$.
\end{definition}

\subsubsection{Kefalopoda Games}
Apart from the trivial inequalities, we identify one family of facet-defining inequalities for the polytope of source digraphs.
These inequalities are represented by \emph{kefalopoda digraphs} (see Fig.~\ref{subfig:kefalopoda}).
\begin{definition}[Kefalopoda digraph]
  For any function~$f:[n]\rightarrow[n]$ without fixed-points, i.e, with~$\forall i:f(i)\neq i$, the \emph{kefalopoda~$\kappa_f$} is
  \begin{equation}
    \mathcal V(\kappa_f)=[n]
    \,,
    \qquad
    \mathcal A(\kappa_f)
    =
    \{
      f(i)\arc i \mid i\in[n]
    \}
    \,.
  \end{equation}
\end{definition}
The term \emph{kefalopoda} is Greek and stands for the family of marine animals that are characterized by a head and any number of attached feet, such as the octopus.
We believe that this term best describes these digraphs.
The reason is that a kefalopoda consists of one or more connected components, where each component has exactly one directed cycle (head), with any number of attached out-trees (feet).
Kefalopoda digraphs are exactly those digraphs where each vertex has in-degree~$1$, and~$f$ describes that \emph{predecessor.}

\begin{theorem}[Source-digraph facets]
  \label{thm:sdfacets}
  The following inequalities are facet-defining for the source-digraph polytope~$\mathcal Q_n^\text{src}$:
  \begin{align}
    & (- \bm 1_{i,j}, 0) & \text{(for any $(i,j)\in[n]^2_{\neq}$; trivial)}
    \,,
    \label{eq:sdfacetstrivial1}
    \\
    & (\bm 1_{i,j}, 1) & \text{(for~$n\geq 3$ and any $(i,j)\in[n]^2_{\neq}$; trivial)}
    \,,
    \label{eq:sdfacetstrivial2}
    \\
    & (\bm \alpha(\kappa_f), n-1) & \text{(for any fixed-point-free function~$f:[n]\rightarrow[n]$)}
    \,.
    \label{eq:sdfacetskefalopoda}
  \end{align}
\end{theorem}
\begin{proof}
  The polytope~$\mathcal Q_n^\text{src}$ is full-dimensional because the set~$\mathcal T$ of the digraphs with at most one arc consists of~$d+1$ affinely independent source digraphs.
  All inequalities are obviously valid for~$\mathcal Q_n^\text{src}$.
  Now, for each family of inequalities, we construct a set of~$d$ affinely independent source digraphs that saturate the respective inequality; hence, they are facet-defining.
  For the inequality~\eqref{eq:sdfacetstrivial1}, this set is~$\mathcal T$ without the digraph with the arc~$i\arc j$.
  For the inequality~\eqref{eq:sdfacetstrivial2} with~$n\geq 3$, this is the set of digraphs with the arc~$i\arc j$ and at most one additional arc.
  Finally, for the kefalopoda inequality~\eqref{eq:sdfacetskefalopoda}, the set consists of all digraphs~$D_{i,j}$ with
  \begin{equation}
    \mathcal V(D_{i,j})=[n]
    \,,
    \qquad
    \mathcal A(D_{i,j})
    =
    \{i \arc j\}
    \cup
    \mathcal A(\kappa_f)
    \setminus
    \{f(i) \arc i\}
  \end{equation}
  for~$i,j\in[n]^2_{\neq}$:
  Take an arbitrary arc $i\arc j$, all arcs from the kefalopoda digraph, and make $i$ a source by removing the single arc pointing to $i$.
  The digraph~$D_{i,j}$ saturates the inequality because it has~$n-1$ arcs in common with~$\kappa_f$.
  Also,~$D_{i,j}$ is a source digraph with the source~$i$.
  The cardinality of this set is $n(n-1)$ because for every $i\arc j$ there is one such digraph $D_{i,j}$ in the set.
  Finally, these digraphs are affinely independent for the following reason.
  The digraph $D_{i,j}$ with $i\arc j\not\in \mathcal A(\kappa_f)$ is the only digraph that contributes in the dimension $(i,j)$.
  Similarly, the digraph~$D_{i,j}$ with~$i\arc j\in\mathcal A(\kappa_f)$ is the only digraph with a 0 at coordinate $(i,j)$.
\end{proof}

Again, we are in position to apply the lifting lemma, and we get:
\begin{corollary}[Definite causal order]
  \label{cor:kefalopoda}
  The kefalopoda inequality
  \begin{align}
    & \graphicaltest\left(\kappa_f, 1 - \frac{1}{2n}\right) & \text{(for any fixed-point-free function~$f:[n]\rightarrow[n]$)}
  \end{align}
  is a $n$-party graphical test for definite causal order.
  Moreover, this graphical test is facet-defining for~$\mathcal P_n^\text{causal}$.
\end{corollary}
For~$n=3$ parties, there are only two non-isomorphic kefalopoda digraphs: the~$3$-cycle and the~$2$-cycle with one outgoing arc attached.
Under the constraint of definite causal order, the chance of winning these games is upper bounded by~$5/6$.
In Fig.~\ref{fig:kefatest} we present all kefalopoda tests for definite causal order for~$n=4$ parties.
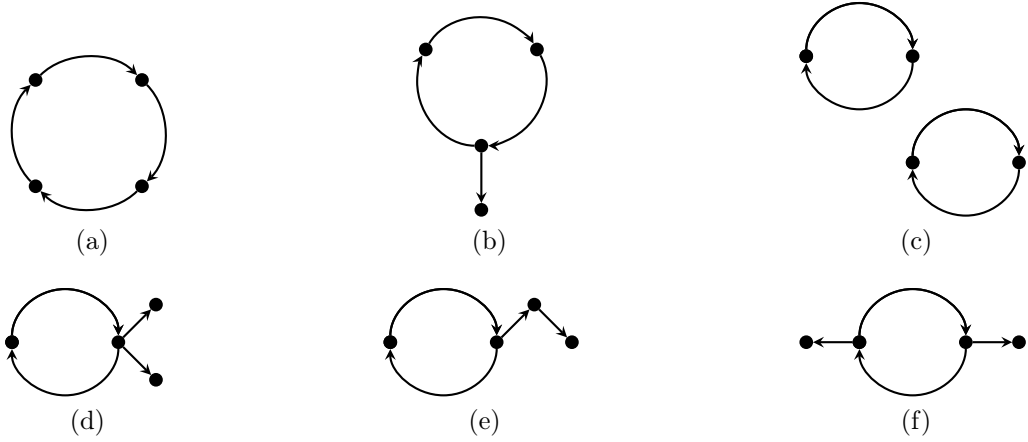
\begin{figure}
  \subfloat[\label{subfig:4cycle}]{%
    \begin{tikzpicture}
      \def\tot{4}
      \pgfmathsetmacro{\angledelta}{360/\tot}
      \foreach \x in {1, 2, 3, ..., \tot} {
        \pgfmathsetmacro{\angle}{(\x-1) * \angledelta + \angledelta/2 + 90}
        \fill (\angle:\radiusmain) circle (\vertexradius);
        \draw[arc] (\angle:\radiusmain) arc[start angle=\angle, delta angle=-\angledelta,radius=\radiusmain];
      }
    \end{tikzpicture}
  }
  \hfill
  \subfloat[\label{subfig:3cycle1}]{%
    \begin{tikzpicture}
      \def\tot{3}
      \pgfmathsetmacro{\radius}{sqrt(pow(\radiusmain,2)/2)/2+\radiusmain/2}
      \pgfmathsetmacro{\angledelta}{360/\tot}
      \foreach \x in {1, 2, 3, ..., \tot} {
        \pgfmathsetmacro{\angle}{(\x-1) * \angledelta + \angledelta/2 + 90}
        \fill (\angle:\radius) circle (\vertexradius);
        \draw[arc] (\angle:\radius) arc[start angle=\angle, delta angle=-\angledelta,radius=\radius];
      }
      \pgfmathsetmacro{\angle}{(2-1) * \angledelta + \angledelta/2 + 90}
      \fill (\angle:\radius) -- ++(0,-\radius) coordinate (p) circle (\vertexradius);
      \draw[arc] (\angle:\radius) -- (p);
    \end{tikzpicture}
  }
  \hfill
  \subfloat[\label{subfig:22cycle}]{%
    \begin{tikzpicture}
      \def\tot{2}
      \pgfmathsetmacro{\radius}{sqrt(pow(\radiusmain,2)/2)}
      \pgfmathsetmacro{\angledelta}{360/\tot}
      \foreach \x in {1, 2, 3, ..., \tot} {
        \pgfmathsetmacro{\angle}{(\x-1) * \angledelta + \angledelta/2 + 90}
        \fill (\angle:\radius) circle (\vertexradius);
        \draw[arc] (\angle:\radius) arc[start angle=\angle, delta angle=-\angledelta,radius=\radius];
      }
      \foreach \x in {1, 2, 3, ..., \tot} {
        \pgfmathsetmacro{\angle}{(\x-1) * \angledelta + \angledelta/2 + 90}
        \path (\angle:\radius) -- ++(-45:2*\radiusmain) coordinate (p);
        \fill (p) circle (\vertexradius);
        \draw[arc] (p) arc[start angle=\angle, delta angle=-\angledelta,radius=\radius];
      }
    \end{tikzpicture}
  }

  \subfloat[\label{subfig:2cycletree}]{%
    \begin{tikzpicture}
      \def\tot{2}
      \pgfmathsetmacro{\radius}{sqrt(pow(\radiusmain,2)/2)}
      \pgfmathsetmacro{\angledelta}{360/\tot}
      \foreach \x in {1, 2, 3, ..., \tot} {
        \pgfmathsetmacro{\angle}{(\x-1) * \angledelta + \angledelta/2 + 90}
        \fill (\angle:\radius) circle (\vertexradius);
        \draw[arc] (\angle:\radius) arc[start angle=\angle, delta angle=-\angledelta,radius=\radius];
      }
      \pgfmathsetmacro{\angle}{(2-1) * \angledelta + \angledelta/2 + 90}
      \fill (\angle:\radius) -- ++(45:\radius) coordinate (p) circle (\vertexradius);
      \draw[arc] (\angle:\radius) -- (p);
      \pgfmathsetmacro{\angle}{(2-1) * \angledelta + \angledelta/2 + 90}
      \fill (\angle:\radius) -- ++(-45:\radius) coordinate (p) circle (\vertexradius);
      \draw[arc] (\angle:\radius) -- (p);
    \end{tikzpicture}
  }
  \hfill
  \subfloat[\label{subfig:2cyclepath}]{%
    \begin{tikzpicture}
      \def\tot{2}
      \pgfmathsetmacro{\radius}{sqrt(pow(\radiusmain,2)/2)}
      \pgfmathsetmacro{\angledelta}{360/\tot}
      \foreach \x in {1, 2, 3, ..., \tot} {
        \pgfmathsetmacro{\angle}{(\x-1) * \angledelta + \angledelta/2 + 90}
        \fill (\angle:\radius) circle (\vertexradius);
        \draw[arc] (\angle:\radius) arc[start angle=\angle, delta angle=-\angledelta,radius=\radius];
      }
      \pgfmathsetmacro{\angle}{(2-1) * \angledelta + \angledelta/2 + 90}
      \fill (\angle:\radius) -- ++(45:\radius) coordinate (p) circle (\vertexradius);
      \draw[arc] (\angle:\radius) -- (p);
      \fill (p) -- ++(-45:\radius) coordinate (1) circle (\vertexradius);
      \draw[arc] (p) -- (1);
    \end{tikzpicture}
  }
  \hfill
  \subfloat[\label{subfig:2cycle2}]{%
    \begin{tikzpicture}
      \def\tot{2}
      \pgfmathsetmacro{\radius}{sqrt(pow(\radiusmain,2)/2)}
      \pgfmathsetmacro{\angledelta}{360/\tot}
      \foreach \x in {1, 2, 3, ..., \tot} {
        \pgfmathsetmacro{\angle}{(\x-1) * \angledelta + \angledelta/2 + 90}
        \fill (\angle:\radius) circle (\vertexradius);
        \draw[arc] (\angle:\radius) arc[start angle=\angle, delta angle=-\angledelta,radius=\radius];
      }
      \pgfmathsetmacro{\angle}{(2-1) * \angledelta + \angledelta/2 + 90}
      \fill (\angle:\radius) -- ++(0:\radius) coordinate (p) circle (\vertexradius);
      \draw[arc] (\angle:\radius) -- (p);
      \pgfmathsetmacro{\angle}{(1-1) * \angledelta + \angledelta/2 + 90}
      \fill (\angle:\radius) -- ++(180:\radius) coordinate (p) circle (\vertexradius);
      \draw[arc] (\angle:\radius) -- (p);
    \end{tikzpicture}
  }
  \caption{%
    All kefalopoda graphical games for~$n=4$ parties.
    When the parties win any of these games with a chance higher than~$7/8$, then their statistics \emph{disagree} with definite causal order.
  }
  \label{fig:kefatest}
\end{figure}

\subsection{Bi-Causal Order}
Some multi-party correlations may be incompatible with any definite causal order simply because the correlations among few of them are.
The constraint of \emph{bi-causal order,} or rather the violation thereof, captures this \emph{genuinely multi-party character} of correlations incompatible with any definite causal order.
This is similar to the notions of genuinely multi-party nonlocality~\cite{svetlichny1987,gallego2012,bancal2013}.
Here, the parties are partitioned into two non-empty groups.
Within each group communication is unrestricted.
Also, one group (the ``past'') can communicate to the other (the ``future'').
It is only that the parties from the ``future'' group cannot signal to any party in the ``past.''
\begin{definition}[Bi-causal order,~$\mathcal C_n^\text{bi}$]
  \label{def:bicausal}
  The~$n$-party correlations~$P_{A|S,R,X}$ agree with \emph{bi-causal order} if and only if
  their probabilities decompose as
  \begin{equation}
    p(a|s,r,x)
    =
    \sum_{\mathcal F: s\not\in\mathcal F,r\in\mathcal F}
    p(\mathcal F)
    p^{\bot}_{\mathcal F}(a|s,r)
    +
    \sum_{\mathcal F: r\in\mathcal F\rightarrow s\in\mathcal F}
    p(\mathcal F)
    p^{\top}_{\mathcal F}(a|s,r,x)
    \,,
  \end{equation}
  where~$\mathcal F\subsetneq[n]$ is a non-empty strict subset of $[n]$.
  The set of such correlations is~$\mathcal C_n^\text{bi}$.
\end{definition}

\subsubsection{Not-Strong Polytope}
As for the previous two constraints, the set~$\mathcal C_n^\text{bi}$ forms an operational 0/1 polytope~$\mathcal P_n^\text{bi}$.
Again, we use the map~$\Gamma$ to study the resulting set of digraphs.
Here, it turns out that the connectivity of a digraph plays a deceive rôle.
\begin{theorem}[Not-strong digraphs]
  The set of digraphs~$\Gamma(\mathcal P_n^\text{bi})$ is the set of all digraphs on~$n$ vertices that are not strongly connected:
  \begin{equation}
    \Gamma(\mathcal P_n^\text{bi})
    =
    \{ D \mid \mathcal V(D)=[n],\exists (i,j) \in[n]^2_{\neq}: j\not\in\suc_D(i) \}
    \,.
  \end{equation}
\end{theorem}
\begin{proof}
  This proof is similar to the other digraph proofs.
  The~$\subseteq$-inclusion follows from the fact that no party in~$\mathcal F$ can signal to any party in~$[n]\setminus\mathcal F$,
  and therefore no path from~$i_0$ to~$j_0$ exists for any pair~$(i_0\in\mathcal F,j_0\in[n]\setminus\mathcal F)$.
  For the opposite inclusion, let~$D$ be a digraph that is not strong.
  This means we can decompose~$D$ into a family~$\{\mathcal H_i\}_i$ of strongly connected components.
  Now, the correlations
  \begin{equation}
    p(0|s,r,x) =
    \begin{cases}
      1-x & (s,r)\in\mathcal A(D) \\
      1 & \text{otherwise}
    \end{cases}
  \end{equation}
  respect that signaling and are bi-causal with~$\mathcal F=\mathcal H_{i_0}$, where~$\mathcal H_{i_0}$ is a source in the \emph{condensation} of~$D$.
\end{proof}
\begin{definition}[Not-strong polytope, $\mathcal Q_n^\text{notstrong}$]
  The \emph{polytope of not-strong digraphs} is~$\mathcal Q_n^\text{notstrong} := \conv(\Gamma(\mathcal P_n^\text{bi}))$.
\end{definition}

\subsubsection{Minimally Strong Games}
We identify facets of the not-strong polytope~$\mathcal Q_n^\text{notstrong}$ that correspond to minimally strong digraphs~\cite{geller1970}.
These are strong digraphs where the removal of an arbitrary arc breaks the strong connectivity.
\begin{definition}[Minimally strong digraphs~\cite{geller1970}]
  The set~$\mathcal M_n$ is the set of all strong digraphs~$D$ over~$\mathcal V(D)=[n]$
  such that for any~$(i\arc j)\in\mathcal A(D)$, the digraph~$D_{i,j}$ with
  \begin{equation}
    \mathcal V(D_{i,j}) = [n]
    \,,
    \qquad
    \mathcal A(D_{i,j}) =
    \mathcal A(D) \setminus \{i\arc j\}
  \end{equation}
  is \emph{not strongly connected.}
\end{definition}
In Fig.~\ref{subfig:minstrong} we present an example of such a digraph.

\begin{theorem}[Not-strong facets]
  \label{thm:notstrongfacets}
  The following inequalities are facet-defining for the not-strong polytope~$\mathcal Q_n^\text{notstrong}$:
  \begin{align}
    & (- \bm 1_{i,j}, 0) & \text{(for any $(i,j)\in[n]^2_{\neq}$; trivial)}
    \,,
    \label{eq:nsfacetstrivial1}
    \\
    & (\bm 1_{i,j}, 1) & \text{(for~$n\geq 3$ and any $(i,j)\in[n]^2_{\neq}$; trivial)}
    \,,
    \label{eq:nsfacetstrivial2}
    \\
    & (\bm \alpha(G), |\mathcal A(G)|-1) & \text{(for any~$G\in\mathcal M_n$)}
    \,.
    \label{eq:nsfacetsminimal}
  \end{align}
\end{theorem}
\begin{proof}
  The proof of full-dimensionality and trivial inequalities for Theorem~\ref{thm:sdfacets} carries over to the present theorem.
  The minimally strong inequalities~\eqref{eq:nsfacetsminimal} are obviously valid for~$\mathcal Q_n^\text{notstrong}$.
  Now, for some~$G\in\mathcal M_n$, we construct a set of~$d$ affinely independent digraphs that are not strongly connected, and where each of them saturates the inequality~\eqref{eq:nsfacetsminimal}.
  These are constructed by following a similar pattern as for the kefalopoda inequalities~\eqref{eq:sdfacetskefalopoda}.
  First, consider the set
  \begin{equation}
    \mathcal S_0 := \{ D \mid \mathcal V(D)=[n], \mathcal A(D) \subsetneq \mathcal A(G), |\mathcal A(D)| + 1 = |\mathcal A(G)| \}
  \end{equation}
  of all digraphs obtained by the removal of a single arc from~$G$.
  These digraphs are, by definition, not strongly connected.
  Also, the set~$\mathcal S_0$ contains~$|\mathcal A(G)|$ such digraphs, which are affinely independent.
  Next, for each arc~$i\arc j$ absent in~$G$, we construct a digraph~$D_{i,j}$ that saturates the inequality and contains that arc:
  \begin{equation}
    \mathcal V(D_{i,j}) = [n]
    \,,
    \qquad
    \mathcal A(D_{i,j}) = \{i\arc j\} \cup \mathcal A(G) \setminus \{ j\arc k\}
    \,,
  \end{equation}
  where~$j\arc k\in\mathcal A(G)$ is the first arc on a path from~$j$ to~$i$ in~$G$.
  Note that, because~$G$ is strong, such an arc~$j\arc k$ always exists.
  Also,~$D_{i,j}$ is not strongly connected.
  This is because adding the arc~$i\arc j$ does not compensate the removal of~$j\arc k$ from~$G$.
  These digraphs are affinely independent because each one of them contributes to an otherwise untouched coordinate $(i,j)$.
  Altogether, we have~$n(n-1)$ affinely independent digraphs that saturate said inequality.
\end{proof}

Now, we apply the lifting lemma again, and obtain:
\begin{corollary}[Bi-causal order]
  \label{cor:minimal}
  The minimally strong inequality
  \begin{align}
    & \graphicaltest\left(G, 1-\frac{1}{2|\mathcal A(G)|}\right) & \text{(for any~$G\in\mathcal M_n$)}
  \end{align}
  is a $n$-party graphical test for bi-causal order.
  Moreover, this graphical test is facet-defining for~$\mathcal P_n^\text{bi}$.
\end{corollary}
For~$n=3$ parties, there are only two non-isomorphic minimally strong digraphs: the~$3$-cycle and the digraph with two~$2$-cycles.
Under the constraint of bi-causal order, the chance of winning these games is upper bounded by~$5/6$ and~$7/8$, respectively.
In Fig.~\ref{fig:minstrongtest}, we present all minimally strong tests of bi-causal order for~$n=4$ parties.
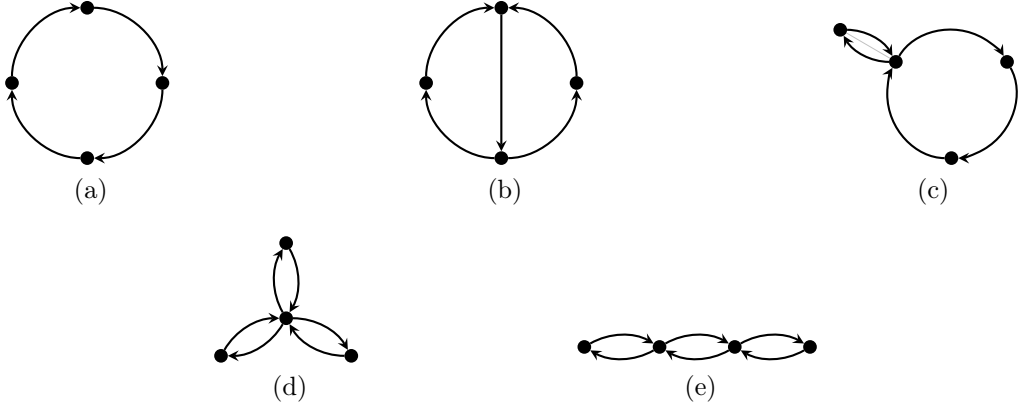
\begin{figure}
  \centering
  \subfloat[\label{subfig:4cyclems}]{%
    \begin{tikzpicture}
      \def\tot{4}
      \pgfmathsetmacro{\angledelta}{360/\tot}
      \foreach \x in {1, 2, 3, ..., \tot} {
        \pgfmathsetmacro{\angle}{(\x-1) * \angledelta + \angledelta/2 + 45}
        \fill (\angle:\radiusmain) circle (\vertexradius);
        \draw[arc] (\angle:\radiusmain) arc[start angle=\angle, delta angle=-\angledelta,radius=\radiusmain];
      }
    \end{tikzpicture}
  }
  \hfill
  \subfloat[\label{subfig:3cycles}]{%
    \begin{tikzpicture}
      \def\tot{4}
      \pgfmathsetmacro{\angledelta}{360/\tot}
      \foreach \x in {1, 2, 3, ..., \tot} {
        \pgfmathsetmacro{\angle}{(\x-1) * \angledelta + \angledelta/2 + 45 + 90}
        \fill (\angle:\radiusmain) circle (\vertexradius);
        \ifnum\x<3
          \draw[arc] (\angle:\radiusmain) arc[start angle=\angle, delta angle=-\angledelta,radius=\radiusmain];
        \else
          \draw[arcrev] (\angle:\radiusmain) arc[start angle=\angle, delta angle=-\angledelta,radius=\radiusmain];
        \fi
      }
      \pgfmathsetmacro{\angleA}{(4-1) * \angledelta + \angledelta/2 + 45 + 90}
      \pgfmathsetmacro{\angleB}{(2-1) * \angledelta + \angledelta/2 + 45 + 90}
      \draw[arc] (\angleA:\radiusmain) -- (\angleB:\radiusmain);
    \end{tikzpicture}
  }
  \hfill
  \subfloat[\label{subfig:32cycle}]{%
    \begin{tikzpicture}
      \def\tot{3}
      \pgfmathsetmacro{\radius}{sqrt(pow(\radiusmain,2)/2)/2 + \radiusmain/2}
      \pgfmathsetmacro{\angledelta}{360/\tot}
      \foreach \x in {1, 2, 3, ..., \tot} {
        \pgfmathsetmacro{\angle}{(\x-1) * \angledelta + \angledelta/2 + 90}
        \fill (\angle:\radius) circle (\vertexradius);
        \draw[arc] (\angle:\radius) arc[start angle=\angle, delta angle=-\angledelta,radius=\radius];
      }
      \pgfmathsetmacro{\angle}{(1-1) * \angledelta + \angledelta/2 + 90}
      \fill (\angle:\radius) -- ++(\angle:\radius) coordinate (a) circle (\vertexradius);
      \draw[arc] (\angle:\radius) to[bend left] (a);
      \draw[arc] (a) to[bend left] (\angle:\radius);
    \end{tikzpicture}
  }

  \hfill
  \subfloat[\label{subfig:32cycles}]{%
    \begin{tikzpicture}
      \def\tot{3}
      \pgfmathsetmacro{\angledelta}{360/\tot}
      \fill (0,0) circle (\vertexradius);
      \foreach \x in {1, 2, 3, ..., \tot} {
        \pgfmathsetmacro{\angle}{(\x-1) * \angledelta + \angledelta/2 - 90}
        \fill (\angle:\radiusmain) circle (\vertexradius);
        \draw[arc] (\angle:\radiusmain) to[bend left] (0,0);
        \draw[arcrev] (\angle:\radiusmain) to[bend right] (0,0);
      }
    \end{tikzpicture}
  }
  \hfill
  \subfloat[\label{subfig:2cycleschain}]{%
    \begin{tikzpicture}
      \fill (0,0) circle (\vertexradius);
      \fill (\radiusmain,0) circle (\vertexradius);
      \draw[arc] (\radiusmain,0) to[bend left] (0,0);
      \draw[arcrev] (\radiusmain,0) to[bend right] (0,0);
      \fill (2*\radiusmain,0) circle (\vertexradius);
      \draw[arc] (2*\radiusmain,0) to[bend left] (\radiusmain,0);
      \draw[arcrev] (2*\radiusmain,0) to[bend right] (\radiusmain,0);
      \fill (3*\radiusmain,0) circle (\vertexradius);
      \draw[arc] (3*\radiusmain,0) to[bend left] (2*\radiusmain,0);
      \draw[arcrev] (3*\radiusmain,0) to[bend right] (2*\radiusmain,0);
    \end{tikzpicture}
  }
  \hfill
  \subfloat{}
  \caption{%
    All minimally strong graphical games for~$n=4$ parties.
    When the parties win the game (a) with a chance higher than~$7/8$,
    or the games (b) -- (c) with a chance higher than~$9/10$,
    or the games (d) -- (e) with a chance higher than~$11/12$,
    then their statistics are \emph{incompatible} with any bi-causal order.
  }
  \label{fig:minstrongtest}
\end{figure}

\subsubsection{An Uncategorised Game}
The ``twisted cylinder'' digraph~$T$ shown in Fig.~\ref{fig:twisteccylinder} is another graphical game for bi-causal order.
\begin{figure}
  \centering
  \begin{tikzpicture}
    \fill (0,\radiusmain) circle (\vertexradius);
    \fill (0,-\radiusmain) circle (\vertexradius);
    \draw[arc] (0,\radiusmain) to[bend left] (0,-\radiusmain);
    \draw[arcrev] (0,\radiusmain) to[bend right] (0,-\radiusmain);
    \fill (5*\radiusmain,\radiusmain) circle (\vertexradius);
    \fill (5*\radiusmain,-\radiusmain) circle (\vertexradius);
    \draw[arc] (5*\radiusmain,\radiusmain) to[bend left] (5*\radiusmain,-\radiusmain);
    \draw[arcrev] (5*\radiusmain,\radiusmain) to[bend right] (5*\radiusmain,-\radiusmain);
    \draw[arc] (0,\radiusmain) -- (5*\radiusmain,\radiusmain);
    \draw[arc] (0,-\radiusmain) -- (5*\radiusmain,-\radiusmain);
    \draw[arcrev] (0,\radiusmain) -- (5*\radiusmain,-\radiusmain);
    \draw[arcrev] (0,-\radiusmain) -- (5*\radiusmain,\radiusmain);
  \end{tikzpicture}
  \caption{%
    This ``twisted cylinder'' is a $4$-party graphical game for bi-causal order.
  }
  \label{fig:twisteccylinder}
\end{figure}
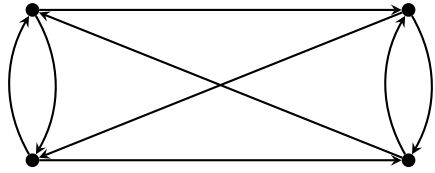
We derived this using \texttt{cddlib}~\cite{cddlib}.
This digraph is strong, but \emph{not} minimally strong:
it remains strong even after removing one of the arcs in a $2$-cycle.
The inequality~$(\bm \alpha(T), 6)$ is facet-defining for~$\mathcal Q_4^\text{bi}$.
Analogously, the inequality
\begin{equation}
  \graphicaltest\left(
    T, \frac{1}{2}+\frac{6}{2|\mathcal A(T)|} = \frac{7}{8}
  \right)
\end{equation}
is facet-defining for~$\mathcal P_4^\text{bi}$.
So far, we have not managed to derive the family of graphical games for bi-causal order~$T$ belongs to.

\subsection{Relations Among Causal Constraints}
\label{subsec:hamilton}
By the definitions of the causal constraints, it is immediate that
\begin{equation}
  \mathcal C_n^\text{static}
  \subseteq
  \mathcal C_n^\text{causal}
  \subseteq
  \mathcal C_n^\text{bi}
  \,.
\end{equation}
Similar inclusion relations hold for the digraph polytopes and the polytopes of single-output correlations.
It is remarkable, however, that the \emph{Hamiltonian-cycle} inequality~$(\bm\alpha(G), n-1)$ with~$G\cong C_n^n$,
is \emph{simultaneously} facet-defining for all three digraph polytopes~$\mathcal Q_n^\text{DAG},\mathcal Q_n^\text{src},\mathcal Q_n^\text{notstrong}$.
Analogously, the tests described by the Hamiltonian cycle are facet-defining for the corresponding correlation polytopes (this was already known~\cite{baumann2025} for~$\mathcal P_n^\text{causal}$),
because the Hamiltonian cycle is itself a kefalopoda and a minimally strong digraph.
Although this Hamiltonian game is insensitive to distinguish among the constraints studied here, in uncovers a ``common denominator.''
Thus, it plays an analogous r\^{o}le to the non-signaling inequalities for the Bell non-local games.
With this in mind, we are inclined to study the \emph{weakest causal constraint} under that common denominator.
We can do this by investigating the digraph polytope that is constraint only by the trivial inequalities and the Hamiltonian-cycle inequalities.
As a preliminary result, we find that for~$n=4$ parties, this Hamiltonian digraph polytope has \emph{non-integer extremal points} (see Fig.~\ref{fig:probabext}, we found those using \texttt{cddlib}~\cite{cddlib}).
\begin{figure}
  \centering
  \subfloat[\label{subfig:probabext1}]{%
    \begin{tikzpicture}
      \fill (0,0) circle (\vertexradius);
      \foreach \angle in {90, 210, 330} {
        \fill (\angle:\radiusmain) circle (\vertexradius);
        \draw[arc,dashed] (\angle:\radiusmain) arc[start angle=\angle, delta angle=-120,radius=\radiusmain];
        \draw[arc] (0,0) to[bend left] (\angle:\radiusmain);
        \draw[arcrev] (0,0) to[bend right] (\angle:\radiusmain);
      }
    \end{tikzpicture}
  }
  \hfill
  \subfloat[\label{subfig:probabext2}]{%
    \begin{tikzpicture}
      \fill (0,0) circle (\vertexradius);
      \foreach \angle in {90, 210, 330} {
        \fill (\angle:\radiusmain) circle (\vertexradius);
        \draw[arc,dashed] (\angle:\radiusmain) arc[start angle=\angle, delta angle=-120,radius=\radiusmain];
        \draw[arc] (0,0) to[bend left] (\angle:\radiusmain);
        \draw[arcrev] (0,0) to[bend right] (\angle:\radiusmain);
      }
      \draw[arc] (90:\radiusmain) to[bend right] (210:\radiusmain);
    \end{tikzpicture}
  }
  \hfill
  \subfloat[\label{subfig:probabext3}]{%
    \begin{tikzpicture}
      \fill (0,0) circle (\vertexradius);
      \foreach \angle in {90, 210, 330} {
        \fill (\angle:\radiusmain) circle (\vertexradius);
        \draw[arc,dashed] (\angle:\radiusmain) arc[start angle=\angle, delta angle=-120,radius=\radiusmain];
        \draw[arc,dashed] (\angle:\radiusmain) to[bend right] (\angle+120:\radiusmain);
        \draw[arc] (0,0) to[bend left] (\angle:\radiusmain);
        \draw[arcrev] (0,0) to[bend right] (\angle:\radiusmain);
      }
    \end{tikzpicture}
  }
  \caption{%
    Examples of non-integer extremal points of the Hamiltonian digraph polytope.
    The dashed arcs have weight~$1/2$, and the solid ones weight~$1$.
    The total weight of the arcs along any Hamiltonian cycle never exceeds~$3$.
  }
  \label{fig:probabext}
\end{figure}
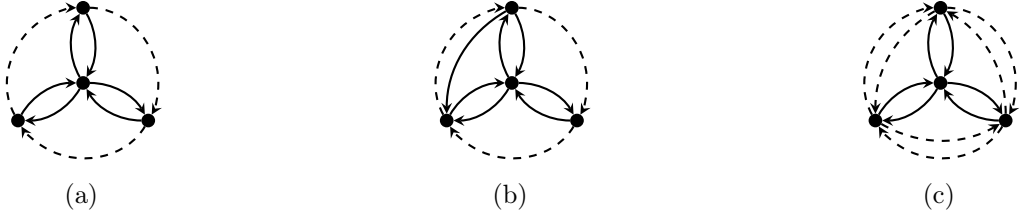
Also, these extremal points
violate a~$2$-cycle inequality (they contain~$C_2^2$),
violate a kefalopoda inequality (they contain the kefalopoda shown in Fig.~\ref{subfig:2cycletree}),
and violate a minimally strong inequality (they contain the minimally strong digraph shown in Fig.~\ref{subfig:32cycles}).
All these violations are maximal.
This digraph may play an analogous rôle to the PR box~\cite{pr1994} in the field of non-locality.
We leave this for future research.

\section{Weakly Causal Correlations}
The kefalopoda inequalities are possibly the only nontrivial facet-defining inequalities of the source-digraph polytope.
This is in stark contrast to the other digraph polytopes studied here, i.e., the DAG polytope and the not-strong polytope.
There, we \emph{know} that additional facet-defining inequalities exist in general; the sets of facets given by Theorem~\ref{thm:dagfacets} and Theorem~\ref{thm:notstrongfacets} are incomplete.
This observation motivates us to turn Theorem~\ref{thm:sdfacets} into a definition of a \emph{weak} version of the source polytope.
This then then translates to a causal constraint, the constraint of \emph{weakly causal correlations.}
We then present an efficient algorithm to decide whether some single-output correlations~$P_{A|S,R,X}$ are weakly causal or not.
\begin{definition}
  The polytope~$\tilde{\mathcal Q}_n^\text{src}\subseteq\mathbb R^{n(n-1)}$ is the set of vectors that satisfy the following inequalities:
  \begin{align}
    & (- \bm 1_{i,j}, 0) & \text{(for any $(i,j)\in[n]^2_{\neq}$; trivial)}
    \,,
    \\
    & (\bm 1_{i,j}, 1) & \text{(for~$n\geq 3$ and any $(i,j)\in[n]^2_{\neq}$; trivial)}
    \,,
    \\
    & (\bm \alpha(\kappa_f), n-1) & \text{(for any fixed-point-free function~$f:[n]\rightarrow[n]$)}
    \,.
  \end{align}
\end{definition}
We know that~$\tilde{\mathcal Q}_n^\text{src}=\mathcal Q_n^\text{src}$ for~$n\leq 4$.
In general, this identity is unsolved.\footnote{%
  If for some~$n$ these polytopes differ, then~$\tilde{\mathcal Q}_n^\text{src}$ will have \emph{non-integer} extremal points.
}
We now define the corresponding polytope of correlations.
This polytope is motivated by the lifting lemma (Lemma~\ref{lemma:lifting}) and the facet-generation lemma (Lemma~\ref{lemma:generative}).

\begin{definition}
  The polytope~$\mathcal U(\mathcal Q)\subseteq\mathbb R^{2d}$ is the set of vectors that satisfy the following inequalities:
  \begin{align}
    & (- \bm 1_{i}, 0) & \text{(for any $i\in[2d]$; trivial)}
    \,,
    \\
    & (\bm 1_{i}, 1) & \text{(for any~$i\in[2d]$; trivial)}
    \,,
    \\
    & ((\bm\phi \bm w,-\bm \phi \bm w), c) & \text{(for any $(\bm w,c)\in\mathcal H(\mathcal Q)$,~$\bm\phi\in\{-1,+1\}^d$)}
    \,,
  \end{align}
  where~$\mathcal H(\mathcal Q)$ is the set of facet-defining inequality for~$\mathcal Q$.
  The polytope of \emph{$n$-party weakly causal correlations} is~$\tilde{\mathcal P}_n^\text{causal}:=\mathcal U(\tilde{\mathcal Q}_n^\text{src})$.
\end{definition}
Thus, the weakly causal single-output correlations~$P_{A|S,R,X}$ are exactly those that satisfy the kefalopoda inequalities, and all ``rotations'' thereof.
In fact, the signs from~$\bm\phi$ simply alter the game for certain sender-receiver pairs,
such that the game is won whenever the receiver outputs the \emph{negated} input given to the sender.
Luckily for us, the polytopes~$\mathcal Q$ and~$\mathcal U(\mathcal Q)$ are in close relation:
There exists a projection of the vectors in~$\mathcal U(\mathcal Q)$, such that deciding membership in that polytope is equivalent
to deciding membership in~$\mathcal Q$ for the projected vector.
This lemma, that we state next, is complementary to the lifting lemma used in the previous parts.
The projection is given by the map~$\Gamma$ (see Eq.,~\eqref{eq:gamma}) extended to non-integer vectors.
\begin{lemma}
  If all nontrivial facet-defining inequalities a polytope~$\mathcal Q\subseteq[0,1]^d$ are non-negative,
  then, whenever~$\bm p=(\bm p^0,\bm p^1)\in[0,1]^{2d}$,
  \begin{equation}
    \bm p \in \mathcal U(\mathcal Q)
    \Longleftrightarrow
    |\bm p^0 - \bm p^1|\in\mathcal Q
    \,,
  \end{equation}
  where the absolute values are taken element-wise.
\end{lemma}
\begin{proof}
  We prove both directions by contradiction.
  For the~$\Leftarrow$ direction, suppose that~$\bm p$ is \emph{not} in~$\mathcal P:=\mathcal U(\mathcal Q)$.
  This means there exists a \emph{violated} lifted inequality~$((\bm\phi\bm w,-\bm\phi\bm w),c)$, i.e.,
  \begin{equation}
    (\bm\phi\bm w,-\bm\phi\bm w)\cdot\bm p > c
    \,,
  \end{equation}
  where~$(\bm w,c)$ is a nontrivial and non-negative facet-defining inequality of~$\mathcal Q$, and~$\bm\phi\in\{-1,+1\}^d$.
  The left-hand side of this equation is
  \begin{equation}
    \sum_{i\in[d]}\bm\phi_i \bm w_i\left(\bm p^0_i - \bm p^1_i\right)
    \leq
    \sum_{i\in[d]}\bm w_i\left|\bm p^0_i - \bm p^1_i\right|
    =
    \bm w\cdot \left|\bm p^0 - \bm p^1\right|
    \,.
  \end{equation}
  Therefore, the projected vector is also not in~$\mathcal Q$.
  For the opposite direction, suppose that~$|\bm p^0 - \bm p^1|$ for some~$\bm p\in[0,1]^{2d}$ is \emph{not} in~$\mathcal Q$.
  Again, this means there exists a nontrivial and non-negative inequality~$(\bm w,c)$ that is violated:
  \begin{equation}
    \bm w \cdot \left|\bm p^0 - \bm p^1\right| > c
    \,.
  \end{equation}
  Now, the can express the left-hand side as
  \begin{equation}
    \sum_{i\in[d]}
    \bm w_i \left|\bm p^0_i - \bm p^1_i\right|
    =
    \sum_{i\in[d]}
    \bm\phi_i \bm w_i \left(\bm p^0_i - \bm p^1_i\right)
    =
    (\bm\phi\bm w,-\bm\phi\bm w)\cdot \bm p
    \,,
  \end{equation}
  where we define~$\bm\phi_i := 1$ if~$\bm p^0_i \geq \bm p^1_i$, and~$\bm \phi_i:=-1$ otherwise.
  Thus, also~$\bm p$ is not in~$\mathcal P$.
\end{proof}

\subsection{Efficient Digraph Algorithm}
By the above lemma, the problem of deciding whether some single-output correlations belong to the polytope of weakly causal correlations is reduced
to the problem whether some vector belongs to~$\tilde{\mathcal Q}_n^\text{src}$.
Towards an algorithm and complexity-theoretic treatment, we need to specify how the input is encoded.
Machines cannot treat real numbers.
So, in the following, we always consider~$d$-dimensional vectors of rational numbers.
Moreover, each rational number~$a/b$ is given as a pair~$(a,b)$ where the length of the binary encodings of~$a$ and~$b$ are upper bounded by~$p(d)$ for fixed polynomial~$p$.
We use~$\langle\bm v\rangle\in\mathbb Q^d_\text{poly}$ to denote the encoding of such a vector.
The decision problem for these digraphs is thus:

\decisionproblem{WeakSourceDigraph}{$n\geq 2$, $\langle\bm q\rangle\in\mathbb Q^{n(n-1)}_\text{poly}$}{Decide whether or not~$\bm q\in\tilde{\mathcal Q}_n^\text{src}$.}
\begin{theorem}
  The decision problem \textsc{WeakSourceDigraph} is in~$\cc{P}$.
\end{theorem}
\begin{proof}
  We prove this statement by providing an explicit algorithm that runs in polynomial time in~$n$.
  The algorithm is straightforward.
  We interpret~$\bm q$ as the adjacency vector of a weighted digraph, where the entry~$\bm q_{s,r}$ is the weight of the arc~$s\arc r$.
  In the algorithm, we iterate over every vertex~$r\in[n]$.
  For each vertex, we consider all incoming arcs~$\mathcal I_r:=\{\bm q_{s,r}\}_s$.
  If one or more of the incoming arcs has a negative weight or a weight that exceeds~$1$, then we abort and answer ``no.''
  Otherwise, we take the maximum weight~$\max\mathcal I_r$, and add it to a register~$t$ initialized to~$0$.
  If, after considering each vertex, the total weight~$t$ is~$n-1$ or less, then we answer ``yes,'' otherwise we answer ``no.''
  This algorithm is correct for the following reason.
  If~$\bm q$ violates a trivial inequality, then this is detected.
  Otherwise, the value~$t$ computed is
  \begin{equation}
    \max_{f:[n]\rightarrow[n]} \bm\alpha(\kappa_f)\cdot\bm q
    \,,
  \end{equation}
  and the answer is correct again.
  This algorithm performs a linear number of comparisons and additions in the dimension~$n(n-1)$.
\end{proof}

\subsection{Efficient Correlations Algorithm}
Similar as above, we encode the probability vectors as vectors of rationals with polynomially bounded constituents.
The above algorithm then translates immediately to an algorithm for solving the decision problem of whether some single-output correlations~$P_{A|S,R,X}$ are weakly causal or not.
For that, simply perform a projection of the input vector first. 
\decisionproblem{WeaklyCausalCorrelations}{$n\geq 2$, $\langle\bm p\rangle\in\mathbb P^{2n(n-1)}_\text{poly}$}{Decide whether or not~$\bm p\in\tilde{\mathcal P}_n^\text{causal}$.}
\begin{corollary}
  The decision problem \textsc{WeaklyCausalCorrelations} is in~$\cc{P}$.
\end{corollary}

Note that since the dimensions of the polytopes of interest are quadratic
in~$n$, the running times of these algorithms are not only polynomial in the
size of the input, but also polynomial in the number of parties~$n$.
This is radically contrasted by the ``size'' of the polytope when we count the number of nontrivial facets.
The polytope~$\tilde{\mathcal Q}_n^\text{src}$ has~$(n-1)^n$ nontrivial facets; this is the number of fixed-point free functions~$f:[n]\rightarrow [n]$.
Then, taking into account the ``rotations''~$\bm\phi$, we get that the polytope~$\tilde{\mathcal P}_n^\text{causal}$ has~$(2(n-1))^n$ nontrivial facets.

\section{Conclusions and Open Questions}
We derive particularly simple inequalities for static causal order, definite causal order, and bi-causal order.
These inequalities, described by digraphs, serve as device-independent tests of different degrees of causality for any number of parties.
This result is made possible through a tight connection between single-output correlations and directed graphs.
Here, the causal constraints are mirrored in graph-theoretic properties, and the study of correlations is reduced to the study of graphs.
It is immediate that all graphical games for~$n\geq 3$ parties can be won perfectly using classical or quantum processes without definite causal order~\cite{baumeler2016space,ocb2012}.
This follows from the Ardehali-Svetlichny type process with genuinely multi-party ``acausality'' derived in Ref.~\cite{baumeler2022}.
This means that winning these games is not at odds with logical consistency.
In fact, the ``Ardehali-Svetlichny'' process is a resource with which \emph{any} graphical game for~$n\geq 3$ parties can be won perfectly.
This, in turn, implies that the single-output scenario studied here seems to be unsuitable for the study of \emph{antinomicity}~\cite{kunjwal2023,kunjwal2024},
which has been proposed as an ultimate level of ``acausality.''\footnote{But then, does antinomicity require the output of two or more parties?}
Nevertheless, as we show, this single-output scenario is sensitive to different degrees of causality.

The present research raises a series of immediate open problems.
One such problem that can be formalized precisely, is:
Is the set of weakly causal correlations the same as the set of causal correlations,
i.e., does the identity~$\mathcal P_n^\text{causal}=\tilde{\mathcal P}_n^\text{causal}$ hold in general?
Another open question is the derivation of all facet-defining inequalities of the bi-causal inequalities (see Fig.~\ref{fig:twisteccylinder}).
The observations on the Hamiltonian-cycle polytope (see Section~\ref{subsec:hamilton}) also initiate further research.
Firstly, does the Hamiltonian game represent some ``essence'' of causal order?
Secondly, is the \emph{converse} of the present study also feasible?
In the present article we derive digraphs from causal constraints.
Instead, one may derive causal constraints from digraphs.
It is possible that such an approach would lead to a fine-grained specification of causal constraints.
This then again, may lead to a hierarchy of computational power and information-processing capabilities in general.
Note also that further causal constraints are described in Ref.~\cite{abbott2017}, that lead up to the bi-causality constraint.
Finally, studying scenarios beyond the single-output scenario may be of interest.
There are two reasons for that.
The first reason is that for two or more outputs, the scenario may become sensitive to quantum-over-classical advantages and the study of antinomicity.
The second reason is that the Hamiltonian game is a \emph{constraint version} of the guess-your-neighbor's-input game~\cite{almeida2010,branciard2015}.
Here, the input of a single neighbour must be guessed.
There, all parties must guess their neighbour's input simultaneously.
This points to a series of intermediate problems that may share similar qualities.

\section*{Acknowledgements}
\noindent
We thank Mateus Araújo for helpful discussions, and
the organizers of the 2nd FoQaCiA workshop, i.e., Rui Soares Barbosa, Anne Broadbent, Ernesto Galvão, and Ana Belén Sainz,
for the opportunity to present parts of these results at that workshop.
This work is supported by the Swiss National Science Foundation (SNF) through project~214808, and
by the Program of Concerted Research Actions (ARC) of the Universit\'e libre de Bruxelles and from the F.R.S.-FNRS under project CHEQS within the Excellence of Science (EOS) program.

\bibliographystyle{elsarticle-num}
\bibliography{refs.bib}
\end{document}